\documentclass[reqno,11pt,a4paper,twoside]{amsart}

\usepackage{amsmath,amssymb,amsthm}
\usepackage[UKenglish]{babel}
\usepackage{ifthen}
\usepackage{mathrsfs}
\usepackage[pdftitle={Quasifree stochastic cocycles and quantum random walks},
pdfauthor={Alexander C. R. Belton, Michal Gnacik, J. Martin Lindsay and P. Zhong},
bookmarks=true,
colorlinks=false,
pdfstartview=FitH]{hyperref}

\usepackage{geometry}
\usepackage{color}

\theoremstyle{plain}
\newtheorem{propn}{Proposition}[section]
\newtheorem{thm}[propn]{Theorem}
\newtheorem{lemma}[propn]{Lemma}
\newtheorem{cor}[propn]{Corollary}

\theoremstyle{definition}
\newtheorem{defn}[propn]{Definition}
\newtheorem{example}[propn]{Example}
\newtheorem{ass}[propn]{Assumption}

\theoremstyle{remark}
\newtheorem*{rem}{Remark}
\newtheorem*{rems}{Remarks}
\newtheorem*{notn}{Notation}

\newtheorem*{motivation}{Motivation}
\newtheorem*{notation}{Notation and conventions}

\newcommand{\indf}[1]{1_{#1}}
\newcommand{\smallgaussintegrand}[3]%
{\left[\begin{smallmatrix}
#1 & #3 \\[0.5ex]
#2 & 0 \end{smallmatrix}\right]}
\newcommand{\smallQFintegrand}[3]%
{\left[\begin{smallmatrix}
#1 & #3 \\[0.5ex]
#2 & 0 \end{smallmatrix}\right]}
\newcommand{\smallQSintegrand}[4]{\left[\begin{smallmatrix}
#1 & #3 \\[0.5ex]
#2 & #4 \end{smallmatrix}\right]}
\newcommand{\smallmat}[4]{\left[\begin{smallmatrix}
#1 & #3 \\[0.5ex]
#2 & #4 \end{smallmatrix}\right]}
\newcommand{\smallcol}[2]{\left[\begin{smallmatrix}
#1 \\[0.5ex]
#2 \end{smallmatrix}\right]}
\newcommand{\smallrow}[2]{\left[\begin{smallmatrix}
#1 &
#2 \end{smallmatrix}\right]}
\newcommand{\gaussintegrand}[3]%
{\begin{bmatrix}
#1 & #3 \\[0.5ex]
#2 & 0 \end{bmatrix}}
\newcommand{\mat}[4]%
{\begin{bmatrix}
#1 & #3 \\[0.5ex] 
#2 & #4 \end{bmatrix}}
\newcommand{\col}[2]%
{\begin{bmatrix}
#1 \\[0.5ex]
#2 \end{bmatrix}}
\newcommand{\row}[2]%
{\begin{bmatrix}
#1 &
#2 \end{bmatrix}}
\newcommand{\phiinitrho}{\phi^\init_\rho}
\newcommand{\phibarinitrho}{\ol{\phi}^\init_\rho}
\newcommand{\Bchil}{B_c}
\newcommand{\Bchilbar}{\Bchil}
\newcommand{\Bcparticle}{\Bchil}
\newcommand{\chil}{c}
\newcommand{\chilbar}{\chil}
\newcommand{\Hint}{H_{\mathsf{I}}}
\newcommand{\Hpar}{H_{\mathsf{P}}}
\newcommand{\Hsys}{H_{\mathsf{S}}}
\newcommand{\Htot}{H_{\mathsf{T}}}
\newcommand{\BVCP}{B_{V,C,P}}
\newcommand{\BVCPprime}{B_{V',C',P'}}
\newcommand{\BVCPinverse}{B_{V^*,-VCV^*,VPV^*}}
\newcommand{\MVCP}{M^{V,C,P}}
\newcommand{\SigmaAVCP}{\Sigma^{A,V,C,P}}
\newcommand{\Vac}{\Omega}
\newcommand{\vac}{\omega}
\newcommand{\Weyl}{W}
\newcommand{\aminusplus}{a^\mp}
\newcommand{\aplusminus}{a^\pm}
\newcommand{\aplus}{a^{+}}
\newcommand{\aminus}{a^{-}}
\newcommand{\WeylH}{\Weyl_\Hil}
\newcommand{\WeylHbar}{\Weyl_\Hilbar}
\newcommand{\WeylSigma}{\Weyl_\Sigma}
\newcommand{\VacH}{\Vac_\Hil}
\newcommand{\VacHbar}{\Vac_{\ol{\Hil}}}
\newcommand{\vphiH}{\vphi_\Hil}
\newcommand{\vphiSigma}{\vphi_\Sigma}
\newcommand{\IH}{I_\Hil}
\newcommand{\CCRH}{CCR(\Hil)}
\newcommand{\WeylHplusHbar}{\Weyl_{\Hil \op \Hilbar}}
\newcommand{\VacHplusHbar}{\Vac_{\Hil \op \Hilbar}}
\newcommand{\RH}{R_\Hil}
\newcommand{\RHbar}{R_{\Hilbar}}
\newcommand{\RSigma}{R_\Sigma}
\newcommand{\aplusH}{\aplus_\Hil}
\newcommand{\aminusH}{\aminus_\Hil}

\newcommand{\aplusminusH}{\aplusminus_\Hil}
\newcommand{\aminusplusHbar}{\aminusplus_\Hilbar}
\newcommand{\aplusminusHplusHbar}{\aplusminus_{\Hil\op\Hilbar}}
\newcommand{\aminusplusHplusHbar}{\aminusplus_{\Hil\op\Hilbar}}
\newcommand{\aplusHplusHbar}{\aplus_{\Hil\op\Hilbar}}
\newcommand{\aminusHplusHbar}{\aminus_{\Hil\op\Hilbar}}
\newcommand{\aplusSigma}{\aplus_\Sigma}
\newcommand{\aminusSigma}{\aminus_\Sigma}
\newcommand{\aplusminusSigma}{\aplusminus_\Sigma}
\newcommand{\Domain}{\mathcal{D}}
\newcommand{\Torus}{\mathbb{T}}
\newcommand{\aqf}{\mathfrak{a}}
\newcommand{\aqfSigma}{\mathfrak{a}_\Sigma}
\newcommand{\Hilbar}{{\ol{\Hil}}}
\newcommand{\hilbar}{{\ol{\hil}}}
\newcommand{\fbar}{\ol{f}}
\newcommand{\xbar}{\ol{x}}
\newcommand{\ybar}{\ol{y}}
\newcommand{\zbar}{\ol{z}}
\newcommand{\Abar}{\ol{A}}
\newcommand{\Lbar}{\ol{L}}
\newcommand{\Sbar}{\ol{S}}
\newcommand{\Gbox}{G^{\square}}
\newcommand{\Gsbox}{G_s^{\square}}
\newcommand{\Hsbox}{H_s^{\square}}
\newcommand{\SigmaA}{\Sigma_A}
\newcommand{\SigmaAB}{\Sigma_{A,B}}
\newcommand{\Sigmaiota}{\Sigma \circ \iota}
\newcommand{\Stepk}{\Step_\noise}
\newcommand{\StepK}{\Step_\Kil}
\newcommand{\Exp}{\mathcal{E}}
\newcommand{\ExpK}{\Exp_\Kil}
\newcommand{\Khat}{{\wh{\Kil}}}
\newcommand{\xhat}{{\wh{x}}}
\newcommand{\yhat}{{\wh{y}}}
\newcommand{\fhat}{{\wh{f}}}
\newcommand{\ghat}{{\wh{g}}}
\newcommand{\gbar}{{\ol{g}}}
\newcommand{\uef}{u\ve(f)}
\newcommand{\veg}{v\ve(g)}
\newcommand{\uefgbar}{u\ve(f, \gbar)}
\newcommand{\I}{\mathbb{I}}
\newcommand{\ExpectVac}{\mathbb{E}_\Omega}
\newcommand{\shiftK}{\sigma^\Kil}
\newcommand{\YF}{Y^F}
\newcommand{\YFstar}{Y^{F^*}}
\newcommand{\YFtilde}{Y^{\Ftilde}}
\newcommand{\Utilde}{\wt{U}}
\newcommand{\Ftilde}{\wt{F}}
\newcommand{\Ltilde}{\wt{L}}
\newcommand{\Ktilde}{\wt{K}}
\newcommand{\Htilde}{\wt{H}}
\newcommand{\Gtilde}{\wt{G}}
\newcommand{\Wtilde}{\wt{W}}
\newcommand{\Qtilde}{\wt{Q}}
\newcommand{\Rtilde}{\wt{R}}
\newcommand{\Ttilde}{\wt{T}}
\newcommand{\Atilde}{\wt{A}}
\newcommand{\Sigmatilde}{\wt{\Sigma}}
\newcommand{\pitilde}{\wt{\pi}}
\newcommand{\rhotilde}{\wt{\rho}}
\newcommand{\Lindbladian}{\mathcal{L}}
\newcommand{\jhat}{\wh{\jmath}}
\newcommand{\jnoise}{j^\noise}
\newcommand{\jKil}{j^\Kil}
\newcommand{\kbar}{\ol{\noise}}
\newcommand{\AplusSigma}{A^+_\Sigma}
\newcommand{\AminusSigma}{A^-_\Sigma}
\newcommand{\pT}{\mathsf{T}}
\newcommand{\pc}{\mathsf{c}}
 \newcommand{\tauscale}{\mathit{s}_\tau}
\newcommand{\ejk}{e_{j,k}}
\newcommand{\Huv}{H^u_v}
\newcommand{\ri}{\mathrm{i}}
\newcommand{\sinhAbar}{\ol{\sinh A}}
\newcommand{\LambdaSigma}{\Lambda^\Sigma}
\newcommand{\particle}{\mathfrak{p}}
\newcommand{\half}{\tfrac{1}{2}}
\newcommand{\kaa}{\kil_{\alpha \alpha}}
\newcommand{\kab}{\kil_{\alpha \beta}}
\newcommand{\ka}{\kil_\alpha}
\newcommand{\kb}{\kil_\beta}
\newcommand{\Pa}{P_{\alpha}}
\newcommand{\Pab}{P_{\alpha \beta}}
\newcommand{\Pba}{P_{\beta \alpha}}
\newcommand{\eia}{e^i_\alpha}
\newcommand{\ejb}{e^j_\beta}
\newcommand{\eijab}{e^{ij}_{\alpha \beta}}

\newcommand{\ve}{\varepsilon}
\newcommand{\vp}{\varpi}
\newcommand{\vphi}{\varphi}
\newcommand{\vrho}{\varrho}

\newcommand{\Hil}{\mathsf{H}}
\newcommand{\hil}{\mathsf{h}}
\newcommand{\hilone}{\hil_1}
\newcommand{\hiltwo}{\hil_2}

\newcommand{\Kil}{\mathsf{K}}
\newcommand{\kil}{\mathsf{k}}

\newcommand{\init}{\mathfrak{h}}
\newcommand{\noise}{\mathsf{k}}
\newcommand{\khat}{{\wh{\noise}}}
\newcommand{\Fock}{\mathcal{F}}
\newcommand{\Step}{\mathbb{S}}
\newcommand{\rd}{\mathrm{d}}

\newcommand{\Real}{\mathbb{R}}
\newcommand{\Rplus}{{\Real_+}}
\newcommand{\Comp}{\mathbb{C}}
\newcommand{\Zplus}{\mathbb{Z}_+}

\newcommand{\ip}[3][]{#1\langle #2, #3 #1\rangle}
\newcommand{\norm}[1]{\lVert #1 \rVert}
\newcommand{\bra}[1]{\langle #1 \vert}
\newcommand{\ket}[1]{\vert #1 \rangle}
\newcommand{\dyad}[2]{\ket{#1}\bra{#2}}

\newcommand{\sa}{{\text{\tu{sa}}}}

\newcommand{\wh}{\widehat}
\newcommand{\wt}{\widetilde}
\newcommand{\ol}{\overline}
\newcommand{\ot}{\otimes}
\newcommand{\otalg}{\, \, \underline{\ot}\, \,}
\newcommand{\uwot}{\overline{\ot}\,}
\newcommand{\op}{\oplus}
\newcommand{\les}{\leqslant}
\newcommand{\ges}{\geqslant}

\newcommand{\tu}{\textup}

\DeclareMathOperator{\tr}{tr}
\DeclareMathOperator{\Dom}{Dom}
\DeclareMathOperator{\Ran}{Ran}
\DeclareMathOperator{\Ranbar}{\overline{\Ran}}
\DeclareMathOperator{\Lin}{Lin}
\DeclareMathOperator{\Linbar}{\overline{\Lin}}
\DeclareMathOperator{\Ker}{Ker}
\DeclareMathOperator{\id}{id}
\DeclareMathOperator{\re}{Re}
\DeclareMathOperator{\im}{Im}
\DeclareMathOperator{\spec}{spec}
\DeclareMathOperator{\st.}{st.}

\newenvironment{alist}
{

\begin{enumerate}\addtolength{\itemsep}{0.5ex}}
{\end{enumerate}}

\newenvironment{rlist}
{

\begin{enumerate}\addtolength{\itemsep}{0.5ex}}
{\end{enumerate}}

\numberwithin{equation}{section}

\begin{document}

\title[Quasifree stochastic cocycles and quantum random walks]
{Quasifree stochastic cocycles\\and quantum random walks}

\author[Belton, Gnacik, Lindsay and Zhong]{Alexander C.~R.~Belton$^1$, Micha\l\ Gnacik$^2$, J.~Martin Lindsay$^1$ \\
and Ping Zhong$^{1, 3}$}
\address{$^1$Department of Mathematics and Statistics, Lancaster
	University, Lancaster LA1 4YF, United Kingdom;
	email: {\tt a.belton@lancaster.ac.uk},
	{\tt j.m.lindsay@lancaster.ac.uk}}

\address{$^2$School of Mathematics and Physics, Lion Gate Building, Lion Terrace, University of Portsmouth,
	Portsmouth PO1 3HF, United Kingdom;
	email: {\tt michal.gnacik@port.ac.uk}}


\address{$^3$Current address: Department of Mathematics and Statistics,
	University of Wyoming, Dept. 3036, 1000 E. University Avenue,
    Laramie, WY 82071-3036, USA; email: {\tt pzhong@uwyo.edu}}



\begin{abstract}
The theory of quasifree quantum stochastic calculus for
infinite-dimensional noise is developed within the framework of
Hudson--Parthasarathy quantum stochastic calculus. The question of
uniqueness for the covariance amplitude with respect to which a given
unitary quantum stochastic cocycle is quasifree is addressed, and
related to the minimality of the corresponding stochastic
dilation. The theory is applied to the identification of a wide class
of quantum random walks whose limit processes are driven by quasifree
noises.
\end{abstract}

\maketitle

\section{Introduction}
Quantum stochastic calculus for gauge-invariant quasifree
representations of the canonical commutation (and anticommutation)
relations was originally developed in the 1980s; see \cite{BSW},
[$\text{HL}_{1,2}$] and \cite{Lfermi}. The possibilities afforded for
semigroup dilation using such a calculus were further developed
in~\cite{App} and~\cite{LiM}, with the latter treatment using a theory
of integral-sum kernel operators. One-dimensional squeezed noise is
analysed in~\cite{HHKKR}, where additive and multiplicative cocycles
over a finite-dimensional quantum probability space are studied and an
It\^o table is generated. Recently, quasifree stochastic calculus has
been extended to the cases of squeezed states and infinite-dimensional
noise [$\text{LM}_{1,2}$]. A key ingredient of the latter theory is a
partial transpose defined on a class of unbounded operators affiliated
to the noise algebra, which defies the failure of complete boundedness
for the transpose.

Use of quasifree stochastic calculus may be preferred to the standard
theory founded by Hudson and Parthasarathy \cite{HuP,Partha} for both
physical and mathematical reasons \cite{HL2}.  On the one hand, it
describes systems which are more physically realistic, at non-zero
temperatures for example. On the other hand, the quasifree theory
boasts a fully satisfactory martingale representation theorem
\cite{HL1,LM1}, in contrast to the standard theory, whose
representation theorem is restricted by regularity assumptions which
seem hard to overcome [$\text{PS}_{1,2}$].

The purpose of this article is twofold. The first is to develop
quasifree stochastic calculus in a simplified form within the
standard theory, restricting to quasifree states with bounded
covariance amplitudes and unitary quantum stochastic cocycles with
norm-continuous vacuum-expectation semigroups
(Sections~\ref{sec: QF SC} and~\ref{sec: uniqueness}). The second is
to give a deeper explanation of the continuous limit of the Hamiltonian
description of a repeated-interactions model at non-zero temperature.
Various limits in a similar setting were investigated by Attal and
Joye in [$\text{AJ}_{1,2}$]. In particular, the paper \cite{AtJ}
describes how the quantum Langevin equation, obtained as
limit of a repeated-interactions model with particles in a thermal
state, is driven by noises satisfying quasifree It\^o product
relations (Section~\ref{sec: QRW}). Those parts relating to the first
objective are written so as to facilitate the second. Our main results
are Theorems~\ref{thm: 5.11} and~\ref{thm: 5.15},
which may be summarised as follows.
From a faithful, normal state $\rho$ on $B(\particle)$,
with the latter viewed as the
particle observable algebra, and a total Hamiltonian $\Htot(\tau)$ of
repeated-interaction form, acting on the tensor product
$\particle \ot \init$ for a system space $\init$, we derive a
gauge-invariant covariance amplitude $\Sigma(\rho)$ and a quantum
stochastic cocycle $Y$ with the following properties: $Y$ satisfies a
quantum Langevin equation of a particular form, with respect to
$\Sigma(\rho)$-quasifree noise, and the scaled quantum random walks
generated by $\Htot(\tau)$ converge to $Y$ as the time-step parameter
$\tau$ converges to $0$.

The quasifree CCR representations that we employ are of Araki--Woods
type, determined by two maps: the doubling map
\[
\iota = \begin{bmatrix} I_{\noise} \\ -k \end{bmatrix}: 
\noise \to \noise \op \kbar; \ x \mapsto \binom{x}{-\xbar},
\]
where $(\kbar, k)$ is the Hilbert space conjugate to the quasifree
noise-dimension space $\noise$, and an operator
\[
\Sigma = \mat{\Sigma^0_0}{\Sigma^1_0}{\Sigma^0_1}{\Sigma^1_1} \in %
B( \noise \op \kbar ) = %
\mat{ B( \noise ) }{ B( \noise; \kbar ) }%
{ B( \kbar; \noise ) }{ B( \kbar ) }
\]
for which the real-linear map $\Sigma \circ \iota$ is symplectic. The
corresponding Weyl operators $W_\Sigma(f)$ act on the double Boson
Fock space
\[
\Gamma\bigl( L^2 ( \Rplus; \noise \op \kbar ) \bigr) = %
\Gamma\bigl( L^2 ( \Rplus; \noise ) \bigr) \ot %
\Gamma\bigl( L^2 ( \Rplus;  \kbar ) \bigr)
\]
in the following manner:
\[
W_\Sigma(f) :=
W( \Sigma \iota(f) ) =
W( \Sigma^0_0 f - \Sigma^0_1 \fbar ) \ot W( \Sigma^1_0 f - \Sigma^1_1 \fbar )
\quad
\text{ for all $f \in L^2( \Rplus; \noise)$},
\]
where $W(g)$ denotes the Fock--Weyl operator with test function $g$,
and the operators $\iota$ and $\Sigma$ are extended to act on
functions pointwise; for example, $( \Sigma^0_1 f )(t) := \Sigma^0_1 f(t)$ for
all $t\in\Rplus$. The symplectic hypothesis ensures that $W_\Sigma$
defines a CCR representation. This class of representations is
sufficiently general to include a range of interesting examples, while
being concrete enough to render the resulting stochastic calculus
straightforward to employ with a minimum of technicalities. Details of
this representation theory are given in Section~\ref{sec: CCR}.

Section~\ref{sec: QSC} collects the relevant results from standard
quantum stochastic analysis, chosen in light of the requirements for
the passage to quasifree stochastic calculus in
Section~\ref{sec: QF SC}. We motivate the definition of quasifree
stochastic integrals by combining the It\^o-type quantum stochastic
integration of simple processes with the realisation of quasifree
creation and annihilation operators in terms of creation and
annihilation operators for the Fock representation,
for the case of finite degrees of freedom. It is notable that
quasifree stochastic integrability is unaffected by squeezing the
state; indeed, the resulting transformation of quasifree integrands
may be viewed as a change-of-variables formula for quasifree
stochastic calculus (Theorem~\ref{thm: 4.1}). Our approach
demonstrates the central r\^ole in the theory played by a partial
conjugation, which constrains the class of admissible integrands when
the noise is infinite dimensional. This corresponds to the
partial-transpose operation at the heart of the general quasifree
stochastic analysis in [LM$_{1,2}$].  Viewing quasifree
integrals as particular cases of standard quantum stochastic integrals
allows us to employ the existing modern quantum stochastic theory
\cite{Lgreifswald} and to avoid any application of
Tomita--Takesaki theory.  While maintaining strict mathematical
rigour, the simplicity of our approach makes it very suitable for
applications.

Some uniqueness questions are addressed in
Section~\ref{sec: uniqueness}. We first show that the
change-of-variables effect of squeezing on quasifree integrals means
that, for present purposes, we may restrict to gauge-invariant
quasifree states. Then the stochastic generators of quasifree
Hudson--Parthasarathy cocycles on an initial Hilbert space $\init$ are
parameterised by triples of operators~$(A, H, Q)$, where
$A \in B(\noise)$ is non-negative, $H \in B(\init)$ is self-adjoint,
and~$Q \in B(\init; \noise \ot \init)$ is
$\noise$-\emph{conjugatable}; see Definition~\ref{defn: 3.30}. The
set of triples that generate the same cocycle is parameterised by a
class of self-adjoint operators in $B(\noise)$.  Uniqueness for
quasifree Hudson--Parthasarathy cocycles inducing a given inner
Evans--Hudson flow~$j$ (Definition~\ref{defn: EH}) is related to the
minimality of $j$, as a stochastic dilation of its vacuum-expectation
semigroup, in the sense of~\cite{Bhat}.

The final section, Section~\ref{sec: QRW}, concerns quantum random
walks and the repeated-interactions model \cite{AtP}. After a brief
summary of the relevant results from the standard theory
of quantum random walks
\cite{B1,BGL}, we extend the example of Attal and Joye in two
directions: to allow infinite-dimensional noise, and
to incorporate an enlarged class of interaction Hamiltonians.
We show that their example is part of the following more general
phenomenon. If the particles in the repeated-interactions model are
in a faithful normal state with density matrix $\vrho$ then the
quantum Langevin equation which governs the limit cocycle $U$ is
driven by a gauge-invariant quasifree noise with covariance amplitude
determined by the state. This is proved under the assumptions that
$\vrho$ enjoys exponential decay of its eigenvalues, and the interaction
Hamiltonian is conjugatable (with respect to the Hilbert space
$\particle$ on which $\vrho$ acts) and has no diagonal part with respect
to the eigenspaces of $\vrho$ (Theorem~\ref{thm: 5.15}). The result also
includes sufficient further conditions, on the matrix components of the
interaction Hamiltonian, for the quasifree noise to be the unique one
within the class for which $U$ is quasifree. The GNS space given by the
particle state splits naturally into mutually conjugate upper-triangular
and lower-triangular parts; this splitting may be viewed as being the
origin of the double Fock space arising in the relevant CCR
representation.

We expect our results to be of interest to researchers in
quantum optics and related fields; the importance of quantum
stochastic calculus to quantum control engineering, for example, is
clearly demonstrated in many of the papers contained in the
collection~\cite{Gough}.
In future work, we intend to explore quantum control theory within
this quasifree framework. For initial results on quasifree filtering,
which show the potential benefit of using squeezed fields for state
restoration, see~\cite{Bouten}.

\begin{notation}
Throughout, the symbol $\hil$, sometimes adorned with primes or
subscripts, stands for a generic Hilbert space; with this
understanding, we usually refrain from saying ``let $\hil$ and $\hil'$
be Hilbert spaces'', \emph{et cetera}. All Hilbert spaces considered
are complex and separable, with inner products linear in their second
argument. The space of bounded operators from $\hil$ to $\hil'$ is
denoted $B(\hil; \hil')$, and $B ( \hil )_{\sa}$, $B(\hil)_+$,
$U( \hil )$ and $B(\hil)^\times$ denote respectively the sets of
self-adjoint and non-negative operators in
$B( \hil ) := B(\hil; \hil)$, and the groups of unitary operators on $\hil$
and operators in $B(\hil)$ with bounded inverse.

A \emph{conjugate} Hilbert space of~$\hil$ is a pair $(\hilbar, k)$
consisting of an anti-unitary operator~$k$ from $\hil$ to
a Hilbert space $\hilbar$; this is unique up to isomorphism in the
natural sense. For any~$x \in \hil$ and $A \in B(\hil)$, the vector
$k x \in \hilbar$ and the operator $k A k^{-1}\in B( \hilbar)$ are
abbreviated to $\xbar$ and~$\Abar$ respectively. The closed linear
span of a subset $S$ of a Hilbert space is denoted $\Linbar S$; the
range of a bounded operator $T$ and its closure are denoted
$\Ran T$ and $\Ranbar T$ respectively. The domain of an unbounded
operator $T$ is denoted $\Dom T$.  We employ the Dirac-inspired
\emph{bra} and \emph{ket} notation
\[
\bra{x} : \hil \to \Comp; \ y \mapsto \ip{x}{y} %
\qquad \text{ and } \qquad %
\ket{x} : \Comp \to \hil; \ \lambda \mapsto \lambda x,
\]
for any vector $x \in \hil$.

Algebraic, Hilbert-space and ultraweak tensor products are denoted
$\otalg$, $\ot$ and $\uwot$, respectively. The indicator function of a
set $S$ is denoted $\indf{S}$. The group of complex numbers with unit
modulus is denoted $\Torus$. The integer part of a real number~$r$ is
denoted $\lfloor r \rfloor$.
\end{notation}

\section{CCR representations}
\label{sec: CCR}

In this section, we collect some key facts on CCR representations
and quasifree states. In particular, we introduce the squeezing
matrices and AW~amplitudes that determine the class of quasifree
states that are relevant to us.

Recall that every real-linear operator $T: \hil \to \hil'$ is uniquely
decomposable as $L + A$, where $L$ is complex linear and $A$ is
conjugate linear; $L$ and $A$ are referred to as the \emph{linear} and
\emph{conjugate-linear} parts of $T$. Explicitly,
\begin{equation}
\label{eqn: LA}
Lx := \half \big( Tx - \ri \, T (\ri x) \big)
\quad \text{ and } \quad
Ax := \half \big( Tx + \ri \, T (\ri x) \big)
\qquad
\text{ for all } x \in \hil.
\end{equation}

\begin{defn}
\label{defn: sympl}
A real-linear operator $Z: \hil \to \hil'$ is \emph{symplectic}
if it satisfies
\[
\im \ip{ Zx }{ Zy } = \im \ip{ x }{ y }
\quad
\text{ for all } x,y\in\hil.
\]
We denote the space of
symplectic operators from $\hil$ to $\hil'$ by
$S(\hil; \hil')$,
or
$S(\hil)$ when $\hil' = \hil$,
and the group of symplectic automorphisms of $\hil$ by
$S(\hil)^\times$.
\end{defn}

For a complex linear map $T$ from $\hil$ to $\hil'$, it is easily verified that $T$ is
isometric if and only if it is symplectic. In
particular, $U(\hil)$ is the subgroup of $S(\hil)^\times$
consisting of its complex-linear elements.

It is shown in the appendix that symplectic automorphisms of $\hil$
are automatically bounded. Thus $S ( \hil )^\times$ is a subgroup
of the group of bounded invertible real-linear operators on~ $\hil$.

A parameterisation $B = \BVCP$ for the elements of
$S(\hil)^\times$ is also given in the appendix.

For the rest of this section, we fix a Hilbert space $\Hil$ and let
$(\Hilbar, K)$ be its conjugate Hilbert space.

\subsection*{Fock space}
As emphasised by Segal \cite{Segal}, the Boson Fock space over $\Hil$
has two interpretations, particle and wave:
\[
\Gamma( \Hil ) =
\bigoplus_{n=0}^\infty \Hil^{\vee n} =
\Linbar \{ \ve (x): x \in \Hil \}.
\]
Here $\Hil^{\vee n}$ denotes the $n$th symmetric tensor power of
$\Hil$, with $\Hil^{\vee 0} := \Comp$, and $\ve(x)$ is the exponential
vector corresponding to the test vector $x$:
\[
\ve(x) = ( 1, x, x^{\ot 2} / \sqrt{2!}, \cdots ).
\]
The normalised exponential vector
$\exp ( - \half \norm{x}^2 ) \ve(x)$
is denoted $\vp(x)$,
and the distinguished vector $\ve(0) = \vp(0)$
is denoted $\Omega_\Hil$ and called the Fock \emph{vacuum vector}.
For all $x,y \in \Hil$,
\begin{align*}
\ip{ \ve(x) }{ \ve(y) } = \exp \ip{ x }{ y },
\end{align*}
and the map $\lambda \mapsto \ve(x + \lambda y )$ is holomorphic from
$\Comp$ to $\Gamma( \Hil )$.  As well as being total in
$\Gamma(\Hil)$, the exponential vectors are linearly independent.

For any orthogonal decomposition $ \Hil = \Hil_1 \op \Hil_2$, the
Boson Fock space $\Gamma( \Hil )$ is identified with the tensor
product $\Gamma( \Hil_1 ) \ot \Gamma( \Hil_2 )$ via the natural
isometric isomorphism which sends the exponential vector
$\ve(x_1, x_2)$ to~$\ve(x_1) \ot \ve( x_2)$ for all $x_1 \in \Hil_1$
and~$x_2 \in \Hil_2$.

For any $x \in \Hil$, the
\emph{Fock--Weyl operator} $\WeylH(x)$ is the unique unitary operator
on $\Gamma(\Hil)$ such that
\begin{equation}\label{eqn: Fdef}
\WeylH(x) \vp(y) = \exp( -\ri \im \ip{x}{y} ) \vp(x+y)
\quad \text{ for all } y \in \Hil.
\end{equation}
For all $x,y \in \Hil$,
\begin{subequations}
\label{eqn: Fock}
\begin{align}
&
\label{eqn: Freg}
\text{the map } t \mapsto \WeylH(tx) \vp(y)
\text{ is continuous from $\Rplus$ to }\Gamma(\Hil),\\[1ex]
&
\label{eqn: Ftot}
\Linbar \big\{ \WeylH ( z  ) \VacH: \, z \in \Hil \big\} = \Gamma(\Hil),
\\[1ex]
\text{and} \quad &
\label{eqn: Fvac}   
\ip{ \VacH }{ \WeylH (x) \VacH } = \exp( - \half \norm{x}^2 ).
\end{align}
\end{subequations}

\subsection*{CCR representations}
We let $\CCRH$ denote the universal $C^*$-algebra generated by
unitary elements $\{ w_x : x \in \Hil \}$ satisfying the
canonical commutation relations in Weyl form:
\[
w_x w_y = \exp( -\ri\im \ip{x}{y} ) w_{x+y}
\qquad \text{for all } x, y \in \Hil.
\]
Its existence, uniqueness and simplicity were
established in~\cite{Slawny}. By universality, each
operator $B \in S(\Hil)^\times$,
determines a unique automorphism $\alpha_B$ of $\CCRH$ such that
\[
\alpha_B( w_x ) = w_{B x} \qquad \text{for all } x \in \Hil;
\]
see~\cite{BrR,Petz}. The \emph{gauge transformations} of $\CCRH$ are
the automorphisms induced by the unitary operators on $\Hil$
of the form $x \mapsto \lambda x$, where $\lambda \in \Torus$.

If $W$ is a map from $\Hil$ to $U(\hil)$ satisfying the Weyl form of
the canonical commutation relations, then $W = \pi \circ w$ for a
unique representation $\pi$ of $\CCRH$ on $\hil$.  We therefore often
refer to~$W$ itself as the representation.  A representation $W$ of
$\CCRH$ is \emph{regular} if, for all $x \in \Hil$, the unitary
group $( W(tx) )_{t \in \Real}$ is strongly continuous; in this case,
the Stone generator $R(x)$ of the group is called the \emph{field
operator} corresponding to the test vector~$x$ for the regular
representation~$W$.

\subsection*{Fock representation}
It follows from the definition~(\ref{eqn: Fdef}) and
properties~(\ref{eqn: Freg}) and~(\ref{eqn: Ftot}) that the map
$x \mapsto \WeylH(x)$ defines a regular representation of $\CCRH$ with
cyclic vector $\VacH$; this is called the \emph{Fock representation}.
If $\{ \RH(y) : y \in \Hil\}$ is the corresponding set of field
operators then, for any $x \in \Hil$, the \emph{creation operator}
$\aplusH (x)$ and \emph{annihilation operator} $\aminusH (x)$ are
defined by setting
\[
\aplusH( x ) := %
\half ( \RH( \ri x) + \ri \, \RH( x) ) %
\quad \text{and} \quad %
\aminusH( x ) :=
\half ( \RH( \ri x) - \ri \, \RH( x) ).
\]
They are closed and mutually adjoint operators with common domain
$\Dom \RH(ix) \cap \Dom \RH(x)$, on which the following canonical
commutation relations hold \cite{BrR}:
\[
\norm{ \aplusH(x) \xi }^2 = \norm{ \aminusH(x) \xi }^2 + %
\norm{ x }^2 \norm{ \xi }^2.
\]
For any dense subspace
$\Domain$ of $\Hil$, the subspace
$\Lin \{ \ve(z): z \in \Domain \}$
is a common core for all Fock creation and annihilation operators,
on which their actions are as follows:
\[
\aplusH(x) \ve(z) = \frac{\rd}{\rd t}\ve(z + tx) \Big|_{t=0}
\quad \text{ and } \quad
\aminusH(x) \ve(z) = \ip{x}{z} \ve(z) %
\qquad \text{for all } x,z \in \Hil.
\]

\subsection*{Quasifree states and representations}
Let $\aqf$ be a non-negative real quadratic form on $\Hil$, and
suppose
\begin{equation}
\label{eqn: a sigma} 
\aqf[x] \, \aqf[y] \ges \bigl( \im \ip{x}{y} \bigr)^2
\qquad \text{for all } x, y \in \Hil.
\end{equation}
Then there is a unique state $\vphi$ on $\CCRH$ such that
\begin{equation}
\label{eqn: phi a} 
\vphi( w_x) = \exp \bigl( - \half \aqf[x] \bigr)
\qquad \text{for all } x \in \Hil;
\end{equation}
see~\cite{BrR,Petz}. Being non-negative, the form $\aqf$ polarises to
a symmetric bilinear form \cite{Kur}; in other words, the following
map is real linear in each argument:
\[
\Hil \times \Hil \to \Real; \ (x,y) \mapsto %
\tfrac{1}{4} \big( \aqf[x+y] - \aqf[x-y] \big).
\]
In particular, the following regularity property holds:
for all $x, y \in \Hil$, the map $t \mapsto \aqf[ x + t y ]$ is
continuous on~$\Real$.
If $\dim \Hil < \infty$ then $\aqf$ is bounded and therefore
there exists a bounded non-negative real-linear operator $T$ on $\Hil$
such that $\aqf[x] = \re \ip{ x }{ Tx }$ for all $x \in \Hil$.

\begin{defn}
A state $\vphi$ on $\CCRH$ is said to be (mean zero) \emph{quasifree} if
it satisfies~(\ref{eqn: phi a}) for some
non-negative real quadratic form $\aqf$
satisfying~(\ref{eqn: a sigma}); then $\aqf$ is called the
\emph{covariance of} $\vphi$, and any real-linear operator
$Z: \Hil \to \hil$ such that $\norm { Z x }^2 = \aqf[x]$ for all
$x \in \Hil$ is called a \emph{covariance amplitude} for $\vphi$.

A state $\vphi$ on $\CCRH$ is \emph{gauge invariant} if it is
invariant under each gauge transformation, so that
$\vphi ( w_{\lambda x } ) = \vphi ( w_{ x } )$ for all
$\lambda \in \Torus$ and $x \in \Hil$.
\end{defn}

\begin{rem}
Covariances of gauge-invariant quasifree states on $\CCRH$ are
precisely the complex quadratic forms $\aqf$ on $\Hil$ such that
\begin{equation}
\label{eqn: a norm} 
\aqf[x] \ges \norm{ x }^2
\qquad \text{for all } x \in \Hil.
\end{equation}
\end{rem}

\begin{example}
\label{example: Fock vac state}
The \emph{Fock vacuum state}
$\vphiH$ on $\CCRH$,
given by the identity
\[
\vphiH( w_x ) =
\ip{ \VacH }{ \WeylH(x) \VacH } \qquad \text{for all } x \in \Hil,
\]
is the basic example of a
gauge-invariant quasifree state,
in view of~(\ref{eqn: a norm}) and the identity~(\ref{eqn: Fvac}).
\end{example}

\begin{lemma}
\label{lemma: Zphi}
Let $Z \in S(\Hil; \hil)$. Then $Z$ is a covariance amplitude for a
quasifree state $\vphi$ on~$\CCRH$.  Moreover, if $Z$ is complex
linear then $\vphi$ is gauge invariant.
\end{lemma}

\begin{proof}
The first part follows since
\[
\norm{ Zx } \norm{ Zy } \ges | \ip{ Zx }{ Zy } | \ges %
| \im \ip{ Zx }{ Zy } | = | \im \ip{ x }{ y } |
\qquad \text{for all } x, y \in \Hil.
\]
The second part is immediate.
\end{proof}

\begin{rem}
Proposition~\ref{propn: gi} below shows that a covariance amplitude of
a quasifree state need not be complex linear for the state to be gauge
invariant.
\end{rem}

\begin{defn}
The \emph{doubling map for $\Hil$}
is the following
bounded real-linear operator
defined in terms of its conjugate Hilbert space $(\ol{\Hil}, K)$:
\begin{equation*}
\iota =
\begin{bmatrix} I \\ -K \end{bmatrix} :
\Hil \to \Hil \op \Hilbar,
\quad
x \mapsto \binom{x}{-\xbar}.
\end{equation*}
\end{defn}

Note that the range of the doubling map is total, since
\[
\binom{ x }{ \zbar } = \half %
\big( \iota(x-z) - \ri \, \iota( \ri x + \ri z ) \big) %
\qquad \text{ for all } x, z \in \Hil.
\]

Now set
\begin{equation}
\label{eqn: AWzero}
AW_0(\Hil) := \Big\{ \Sigma = \smallmat{C}{0}{0}{\Sbar}: \
S, C \in B(\Hil)_+, \ S^2 + I_\Hil =  C^2 \Big\} \subseteq %
B( \Hil \op \Hilbar )_+,
\end{equation}
and note that
$AW_0(\Hil) = \big\{ \SigmaA: A \in B( \Hil )_+ \big\}$,
where
\begin{equation*}
\SigmaA :=\begin{bmatrix}
\cosh A & 0 \\0 & \sinhAbar
\end{bmatrix} \in B( \Hil \op \Hilbar )_+.
\end{equation*}

\begin{propn}\label{propn: gi}
Let $\Sigma \in AW_0(\Hil)$. The bounded real-linear operator
$\Sigma \circ \iota$ is symplectic, and the quasifree state on $\CCRH$
with covariance amplitude $\Sigma \circ \iota$ is gauge invariant.

Conversely, let $\vphi$ be a gauge-invariant quasifree state
on~$\CCRH$, the covariance of which is a bounded complex quadratic
form on~$\Hil$. Then $\vphi$ has a covariance amplitude of the
form~$\Sigmaiota$ for a unique operator $\Sigma \in AW_0(\Hil)$.
\end{propn}
\begin{proof}
Let $\Sigma = \smallmat{C}{0}{0}{\Sbar} \in AW_0(\Hil)$, and set
$A := \sinh^{-1} S \in B(\Hil)_+$, so that $\Sigma = \SigmaA$.
Then, for all $x$, $y \in \Hil$,
\[
\ip{ \Sigma \iota(x) }{ \Sigma \iota(y) } = %
\ip{ Cx }{ Cy } + \ip{ \ol{Sx} }{ \ol{Sy} } =
\ip{x}{y} + 2 \re \ip{Sx}{Sy}.
\]
It follows that $\Sigma \circ \iota$ is symplectic, and is therefore a
covariance amplitude of a quasifree state~$\vphi$ on $\CCRH$.  The
resulting covariance $\aqfSigma: x \mapsto \norm{ \Sigma \iota(x) }^2$
satisfies
\begin{equation}\label{eqn: cosh 2}
\aqfSigma[ x ] = \norm{x}^2 + 2 \norm{ Sx }^2 = %
\ip{ x }{ \cosh 2A \, x } \qquad \text{for all } x \in \Hil,
\end{equation}
and is thereby manifestly gauge invariant.

Conversely, let $\aqf$ be the covariance of a gauge-invariant
quasifree state on $\CCRH$ and suppose that $\aqf$ is bounded. Since
$\aqf$ is bounded and such that $\aqf[x] \ges \norm{ x }^2$ for all
$x \in \Hil$, there is a unique operator $R \in B(\Hil)$ such that
$\ip{ x }{ Rx } = \aqf[x]$ for all $x \in \Hil$, and $R \ges \IH$. The
map $A \mapsto \cosh 2A$ is a bijection from $B(\Hil)_+$ onto
$\{ R \in B(\Hil)_+ : \, R \ges I_\Hil \}$, and therefore, by the
identity~(\ref{eqn: cosh 2}), it follows that $\aqf = \aqfSigma$ for a
unique operator $\Sigma = \SigmaA \in AW_0(\Hil)$.
\end{proof}

We now introduce the notion of squeezing, important in quantum optics.
For any $B \in S(\Hil)^\times$, set
\begin{equation*}
M_B := \mat{L}{-K A}{-A K^{-1}}{\Lbar},
\end{equation*}
where $L$ and $A$ are the linear and conjugate-linear parts of $B$.
Thus $M_B \in B( \Hil \op \ol{\Hil} )$.

\begin{propn}\label{propn: MiotaB}
\mbox{}\par
\begin{alist}
\item If $B\in S(\Hil)^\times$ then $M_B$ is the unique operator
$M \in B(\Hil \op \Hilbar)$ such that $M \circ \iota = \iota \circ B$.
\item The map $B \mapsto M_B$ is a faithful representation of the
group $S(\Hil)^\times$ on $\Hil \op \Hilbar$.
\item The map $( A, B ) \mapsto \Sigma_A M_B$ from
$B(\Hil)_+ \times S(\Hil)^\times$ to $B( \Hil \op \Hilbar )$
is injective.
\end{alist}
\end{propn}

\begin{proof}
(a) First note that
\[
M_B \circ \iota = \mat{L}{-KA}{-AK^{-1}}{\Lbar} \col{I}{-K} = %
\col{ L+A }{ -K( A+L ) } = \iota \circ B.
\]
The uniqueness part follows from the totality of $\Ran \iota$.

(b) By definition, the operator $M_{I_\Hil}$ equals $I_{ \Hil \op \Hilbar }$.
It follows from (a) that, for all $B, B' \in S(\Hil)^\times$,
\[
M_B M_{B'} \circ \iota = M_B \circ \iota \circ B' = %
\iota \circ B B' = M_{B B'} \circ \iota,
\]
and so $M_B M_{B'} = M_{B B'}$.
Thus, for each $B \in S(\Hil)^\times$, the operator
$M_B$ is invertible and $( M_B )^{-1} = M_{B^{-1}}$. Furthermore, if
$B$, $B' \in S(\Hil)^\times$ are such that $M_B = M_{B'}$,
then $\iota \circ B = \iota \circ B'$,
so $B = B'$ by the injectivity of $\iota$. Hence (b) holds.

(c) Suppose $( A_1, B_1 )$,
$( A_2, B_2 ) \in B( \Hil )_+ \times S( \Hil )^\times$ are such
that $\Sigma_{A_1} M_{B_1} = \Sigma_{A_2} M_{B_2}$. It follows from
part (b) that $\Sigma_{A_1} = \Sigma_{A_2} M_B$, where
$B = B_2 B_1^{-1}$. Set $C_i = \cosh A_i$ and $S_i = \sinh A_i$,
for~$i = 1$, $2$, and let~$L$ and $A$ be the linear and
conjugate-linear parts of $B$. Then
\[
\mat{C_1}{0}{0}{\ol{S_1}} = %
\mat{C_2}{0}{0}{\ol{S_2}}  \mat{L}{-KA}{-AK^{-1}}{\Lbar} = %
\mat{C_2 L}{-K S_2 A}{-C_2 AK^{-1}}{ \ol{ S_2 L } }.
\]
As $C_2$ and $K$ are invertible, this implies that $A = 0$, so $B$ is
complex linear and thus unitary, and $C_1 = C_2 B$. This implies that
$C_1^2 = C_2 B B^* C_2 = C_2^2$, so $C_1 = C_2$ and $C_1 = C_1 B$. As
$C_1$ is invertible, it follows that $B = I_\Hil$ and (c) holds.
\end{proof}

\begin{defn}\label{defn: AW}
Set
\begin{align*}
& M(\Hil) := \big\{ M_B : B \in S( \Hil)^\times \big\}, \\[1ex]
& AW(\Hil) := \big\{ \Sigma \, M : %
\Sigma \in AW_0(\Hil), M \in M( \Hil) \big\}, \\[1ex]
\text{and} \quad & \SigmaAB := \SigmaA M_B \qquad %
\text{ for all } A \in B(\Hil)_+ \text{ and } %
B \in S(\Hil)^\times.
\end{align*}
We refer to the elements of $M(\Hil)$, $AW(\Hil)$ and $AW_0(\Hil)$
respectively as \emph{squeezing matrices}, \emph{AW~amplitudes} and
\emph{gauge-invariant AW amplitudes} for $\Hil$.
\end{defn}

\begin{rems}
(i)
The AW abbreviation is in acknowledgement of Araki and Woods \cite{ArW}.

(ii)
Each AW amplitude for $\Hil$ is of the form $\SigmaAB$ for a
unique pair $( A, B ) \in B(\Hil)_+ \times S(\Hil)^\times$, by
Proposition~\ref{propn: MiotaB}.

(iii)
Let $\Sigma = \SigmaAB \in AW(\Hil)$. Then $\Sigma \circ \iota$
is symplectic, since it is the composition of symplectic maps
$( \SigmaA \circ \iota ) \circ B$, and so is a covariance amplitude of
a quasifree state on $\CCRH$, by Lemma~\ref{lemma: Zphi}.

(iv)
In terms of the parameterisation
$B = \BVCP := V( \cosh P - C \sinh P )$ of $B \in S(\Hil)^\times$
as in Theorem~\ref{thm: A1}, the squeezing matrices take the following form:
\begin{align}
\label{eqn: squeeze}
M_B & = \MVCP := \mat{V \cosh P }{K V C \sinh P }%
{V C \sinh P \cdot K^{-1} }{\ol{V \cosh P}}, \\[1ex]
( M_B )^{-1} & = M_{B^{-1}}= M^{V^*, -VCV^*, VPV^*} \nonumber
\end{align}
and
\begin{align}
\label{eqn: SUCP}
\SigmaAB & = \SigmaAVCP := %
\mat{\cosh A \cdot V \cosh P}{K \sinh A \cdot V C \sinh P}%
{\cosh A \cdot V C \sinh P \cdot K^{-1}}%
{\ol{ \sinh A \cdot V \cosh P}}.
\end{align}
\end{rems}

\subsection*{Araki--Woods representations}
We are interested in the class of representations $\WeylSigma$ of
$\CCRH$ of Araki--Woods type, and the corresponding quasifree states
$\vphiSigma$, determined by AW amplitudes $\Sigma = \SigmaAB$ as
follows:
\begin{align*}
& \WeylSigma := \WeylHplusHbar \circ \Sigma \circ \iota : %
x \mapsto \WeylHplusHbar \big( \Sigma \iota(x) \big)
\\[1ex]
\text{and } \qquad & \vphiSigma: w_x \mapsto %
\big\langle \VacHplusHbar , \WeylSigma(x) \VacHplusHbar \big\rangle
\qquad (x \in \Hil).
\end{align*}

\begin{rem}
Let $\Sigma = \SigmaAB \in AW(\Hil)$.
On one hand, if $A$ is injective then
$\Ran \Sigma \circ \iota$ is total in $\Hil \op \ol{\Hil}$
from which it follows that $\VacHplusHbar$ is a cyclic vector for the
representation $\WeylSigma$
\cite{Skeide} (see \cite[Proposition 2.1]{Lgreifswald}).
On the other hand, if $A = 0$ then
$\WeylSigma(x) = \WeylH( Bx) \ot I_{\Gamma(\Hilbar)}$ for all
$x \in \Hil$, so
$\Linbar \{ \WeylSigma(x) \VacHplusHbar : x \in \Hil \} = %
\Gamma(\Hil) \ot \VacHbar$.
\end{rem}

These AW representations $\WeylSigma$ inherit regularity from the Fock
representation $\WeylHplusHbar$. As in the Fock case, given any
$x \in \Hil$, setting
\[
\aplusSigma(x) := %
\half \big( \RSigma(\ri x) + \ri \, \RSigma(x) \big) %
\qquad \text{and} \qquad \aminusSigma(x) := %
\half \big( \RSigma(\ri x) + \ri \, \RSigma(x) \big)
\]
defines creation and annihilation operators via the quasifree field
operators $\{ \RSigma(z): z \in \Hil \}$, which are the Stone
generators of the corresponding unitary groups
$( \WeylSigma(t z) )_{t \in \Real}$. We now relate these to Fock
creation and annihilation operators.

Let the AW~amplitude
$\Sigma \in B(\Hil \op \Hilbar )$ have the block-matrix form
$\left[ \begin{smallmatrix}
\Sigma^0_0 & \Sigma^0_1 \\ \Sigma^1_0 & \Sigma^1_1
\end{smallmatrix}\right]$.
The identification
$\Gamma( \Hil \op \Hilbar ) = %
\Gamma( \Hil ) \ot \Gamma( \Hilbar )$
gives that
\[
\WeylSigma(x) = %
\WeylHplusHbar( \Sigma^0_0 x - \Sigma^0_1 \xbar, \Sigma^1_0 x - \Sigma^1_1 \xbar ) = %
\WeylH( \Sigma^0_0 x - \Sigma^0_1 \xbar ) \ot %
\WeylHbar( \Sigma^1_0 x - \Sigma^1_1 \xbar )
\qquad \text{for all } x \in \Hil.
\]
It follows that $\RSigma(x)$ is the closure of the operator
\[
\RH( \Sigma^0_0 x - \Sigma^0_1 \xbar ) \ot I_{\Gamma(\Hilbar)} + %
I_{\Gamma(\Hil)} \ot \RHbar( \Sigma^1_0 x - \Sigma^1_1 \xbar ),
\]
by \cite[Theorem~VIII.33]{ReS}, which implies that
\begin{subequations}
\begin{align}\label{eqn: 2.11a}
\aplusSigma(x) & \supseteq %
\aplusHplusHbar( \Sigma^0_0 x, \Sigma^1_0 x ) + %
\aminusHplusHbar ( \Sigma^0_1 \xbar, \Sigma^1_1 \xbar ) \\[1ex]
\label{eqn: 2.11b}
\text{and} \qquad \aminusSigma(x) & \supseteq %
\aminusHplusHbar( \Sigma^0_0 x, \Sigma^1_0 x ) + %
\aplusHplusHbar ( \Sigma^0_1 \xbar, \Sigma^1_1\xbar ).
\end{align}
\end{subequations}
Thus, in terms of a parameterisation $\Sigma = \SigmaAVCP$,
as in~(\ref{eqn: SUCP}),
\begin{multline*}
\aplusminusSigma(x) \supseteq %
\aplusminusHplusHbar %
\big( \cosh A \cdot U \cosh P \, x, \, %
\ol{ \sinh A \cdot U C \sinh P \, x } \big) \\
 + \aminusplusHplusHbar \big( \cosh A \cdot U C \sinh P \, %
 x, \, \ol{ \sinh A \cdot U \cosh P \, x} \big) \qquad %
\text{for all } x \in \Hil.
\end{multline*}
In particular, for a gauge-invariant AW~amplitude $\Sigma = \SigmaA$,
\[
\aplusminusSigma(x) \supseteq %
\aplusminusH( \cosh A \, x )  \ot I_{\Gamma(\Hilbar)} + %
I_{\Gamma(\Hil)} \ot \aminusplusHbar ( \, \ol{ \sinh A \, x } \, ) %
\qquad \text{for all } x \in \Hil.
\]

\begin{rem}
The absence of minus signs in these relations is due to our choice of
signs in the definition of the doubling map~$\iota$, and the choice of
parameterisation of the symplectic automorphism~$B$.
\end{rem}

\section{Quantum stochastic calculus}
\label{sec: QSC}

In this section we summarise the relevant elements of standard quantum
stochastic calculus \cite{Partha,Meyer,FagPryc,Lgreifswald} in a way
which is adapted to the requirements of the quasifree stochastic
calculus developed in Section~\ref{sec: QF SC}. This section ends with
discussions of the non-uniqueness of implementing quantum stochastic
cocycles for an Evans--Hudson flow, and Bhat's minimality criterion
for quantum stochastic dilations.

For the rest of this article, we fix a Hilbert space~$\init$, which is
referred to as the \emph{initial space} or \emph{system space}. For
this section, we also fix a Hilbert space $\Kil$ as the
\emph{multiplicity space} or \emph{noise dimension space}. In later
sections, this will vary or have further structure.

\begin{notn}\label{notn: 3.1}
We use the abbreviations $\Vac$, $\Weyl$, $\aplus$, $\aminus$ and
$\Fock$ for $\VacH$, $\WeylH$, $\aplusH$, $\aminusH$ and
$\Gamma(\Hil)$, respectively, where the Hilbert space
$\Hil$ equals $L^2( \Rplus; \Kil)$. As is customary, we abbreviate the simple
tensor $u \ot \ve(f)$ to~$u \ve(f)$ whenever $u \in \init$ and
$f \in L^2( \Rplus; \Kil)$.

For each $t \in \Rplus$ we have the decomposition
$\Fock = \Fock_{t)} \ot \Fock_{[t}$, where
$\Fock_{t)} := \Gamma\bigl( L^2( [ 0, t ); \Kil ) \bigr)$
and~$\Fock_{[t} := \Gamma\bigl( L^2( [ t, \infty ); \Kil ) \bigr)$.

The space of compactly supported step functions from $\Rplus$ to
$\Kil$ is denoted~$\Step$. Although we view~$\Step$ as a subspace
of~$L^2( \Rplus; \Kil)$, we always take the right-continuous version
of each step function, thus allowing us to evaluate these functions at
any point in~$\Rplus$.
\end{notn}

Note that $\Step$ enjoys the following useful properties:
\begin{rlist}
\item
If $f \in \Step$ and $t\in \Rplus$ then $1_{[0,t)} f \in \Step$;
\item
the \emph{exponential subspace}
$\Exp := \Lin \{ \ve(f): f \in \Step \}$ is dense in $\Fock$;
\item
the subspace $\Lin\{ f( t ) : t \in \Rplus \}$ is finite dimensional,
for all $f \in \Step$.
\end{rlist}

In what follows we restrict our attention, as much as possible, to
processes composed of bounded operators.

\begin{defn}
\label{defn: 3.3}
An \emph{$\hil$-$\hil'$~process}, or \emph{$\hil$~process} if
$\hil = \hil'$, is a function
\[
X: \Rplus \to B( \hil \ot \Fock; \hil' \ot \Fock ); \ %
t \mapsto X_t
\]
which is \emph{adapted}, so that
\[
X_t \in B( \hil \ot \Fock_{t)}; \hil' \ot \Fock_{t)} ) \otimes %
I_{[t} \qquad \text{for all } t \in \Rplus,
\]
where  $I_{[t}$ is the identity operator on $\Fock_{[t}$, and
\emph{measurable}, so that  the function
\[
\Rplus \to \hil' \otimes \Fock; \ t \mapsto X_t \xi
\]
is weakly measurable for all $\xi \in \hil \ot \Fock$. By
separability, weak measurability may be replaced with strong
measurability here.
$\\$

An $\hil$-$\hil'$~process $X$ is
\begin{rlist}
\item \emph{simple} if it is piecewise constant and right continuous,
so that there exists a strictly increasing sequence
$(t_n)_{n \ge 1} \subseteq \Rplus$ such that $t_1 = 0$ and
$t_n \to \infty$ as $n \to \infty$, with $X$ constant on each interval
$[t_n, t_{n+1})$;
\item \emph{continuous} if $t \mapsto X_t \xi$ is continuous for all
$\xi \in \hil \ot \Fock$;
\item \emph{unitary} if $X_t$ is a unitary operator for all
$t \in \Rplus$.
\end{rlist}

Every $\hil$-$\hil'$~process $X$ has an \emph{adjoint process},
namely the $\hil'$-$\hil$~process $X^*: t \mapsto X_t^*$.
Clearly $X^*$ is simple if $X$ is.
\end{defn}

\begin{notn}
\label{notn: 3.75}
It is convenient to augment the multiplicity space, by setting
\begin{equation*}
\Khat := \Comp \op \Kil, \qquad %
\xhat := \binom{1}{x} \text{ for all } x \in \Kil %
\quad \text{and} \quad \fhat(t) := \wh{f(t)} \text{ for all } %
f \in \Step \text{ and } t \in \Rplus.
\end{equation*}
Thus $\Khat \ot \hil = \hil \op ( \Kil \ot \hil )$
and
any operator $T \in B( \Khat \ot \hil; \Khat \ot \hil' )$
has a block-matrix form
\[
\begin{bmatrix}
T_0^0  & T_1^0 \\[0.5ex] T_0^1 & T_1^1
\end{bmatrix}
\in \begin{bmatrix}
B( \hil; \hil' ) & B( \Kil \otimes \hil; \hil' ) \\[0.5ex]
B( \hil; \Kil \otimes \hil' ) & %
B( \Kil \otimes \hil; \Kil \otimes \hil' )
\end{bmatrix}.
\]
\end{notn}

\begin{rem}
One may also begin with a non-trivial Hilbert space $\Khat$ and, by choosing a
distinguished unit vector $\vac \in \Khat$, obtain $\Kil$ by setting
$\Kil := \Khat \ominus \Comp \vac$. This observation will be useful in
Section~\ref{sec: QRW}.
\end{rem}

\begin{defn}
\label{defn: 3.5}
A \emph{$\Kil$-integrand process on $\init$}, or simply
an \emph{integrand process}, is a $\Khat \ot \init$~process~$F$ such
that, in terms of its block-matrix form $\smallmat{K}{L}{M}{N}$,
\begin{align*}
&
s \to K_s v\ve(g) \text{ and }
s \mapsto M_s \big( g(s) \ot v \ve(g) \big)
\text{ are locally integrable},
\\
\text{ and } & s \to L_s v\ve(g) \text{ and }
s \mapsto N_s \big( g(s) \ot v \ve(g) \big)
\text{ are locally square-integrable},
\end{align*}
for all $v \in \init$ and $g \in \Step$.
\end{defn}

\begin{rem}
Suppose $F$ is a $\Khat \ot \init$~process such that, for all
$x$, $y \in \Kil$, the function
\[
s \mapsto \bigl\| K_s + %
M_s \bigl( \ket{x} \ot I_{\init \ot \Fock } \bigr) \bigr\| + %
\bigl\| ( \bra{y} \ot I_{\init \ot \Fock } )
\bigl( L_s + N_s ( \ket{x} \ot I_{\init \ot \Fock } ) \bigr) \bigr\|^2
\]
is locally integrable. Then $F$ is an integrand process.
\end{rem}

\begin{thm}\label{thm: 3.9}
For any integrand process $F$, there exists a unique family
$\Lambda(F) := ( \Lambda(F)_t )_{t \ges 0}$ of linear operators,
with common domain $\init \otalg \Exp$ and codomain $\init \ot \Fock$,
such that
\begin{equation}
\label{eqn: FFF} 
\ip{ \uef }{ \Lambda(F)_t \veg } = %
\int_0^t \bigl\langle \fhat(s) \ot \uef, %
F_s ( \ghat(s) \ot \veg ) \bigr\rangle \, \rd s
\end{equation}
for all $u$, $v \in \init$, $f$, $g \in \Step$ and $t \in\Rplus$.
Furthermore, if $r$, $t \in \Rplus$ are such that $r \les t$ then
\begin{align*}
\norm{ ( \Lambda(F)_t - \Lambda(F)_r ) \veg }
& \les \int^t_r  \big\| ( K_s + %
M_s ( \ket{g(s)} \ot I_{\init \ot \Fock } ) ) \veg \big\| \, \rd s \\
 & \qquad
 + C(g) \Big\{ \int_r^t \big\|
 \big( L_s + N_s ( \ket{g(s)} \ot I_{\init \ot \Fock } ) \big)
 \veg \big\|^2 \, \rd s %
\Big\}^{1/2}
\end{align*}
for all $u$, $v \in \init$ and $f$, $g \in \Step$,
where $C(g) := \norm{g} + ( 1 + \norm{g}^2 )^{1/2}$.
\end{thm}

\begin{proof}
See \cite[Theorem~3.13]{Lgreifswald}.
\end{proof}

\begin{rem}
The identity (\ref{eqn: FFF}) is known as the
\emph{first fundamental formula} of quantum stochastic calculus.
\end{rem}

\begin{cor}\label{cor: 3.9}
If $F = \smallmat{K}{L}{M}{N}$ is an integrand process and its adjoint
process $F^* = \smallmat{K^*}{M^*}{L^*}{N^*}$ is also an integrand
process then $\Lambda(F^*)_t \subseteq \Lambda(F)^*_t$ for all
$t \in \Rplus$.
\end{cor}

\begin{rem}
\label{rem: 3.12}
If the integrand process $F$ is such that the operator $\Lambda(F)_t$
is bounded, for all $t \in \Rplus$, then taking the closure of each
operator defines a continuous $\init$~process which, by a
slight abuse of notation, we also denote by $\Lambda(F)$.
\end{rem}

\begin{notn}
\label{notn: 3.13}
Let $F = \smallmat{K}{L}{M}{N}$ be an integrand process. Then
\begin{align*}
A^\circ(K) :=& \Lambda \big(
\smallmat{K}{0}{0}{0} \big), \
A^-(M) := \Lambda \big(
\smallmat{0}{0}{M}{0} \big), \ \\
A^+(L) :=& \Lambda \big(
\smallmat{0}{L}{0}{0} \big)
 \text{ and }
A^\times(N) := \Lambda \big(
\smallmat{0}{0}{0}{N} \big)
\end{align*}
are the
\emph{time}, \emph{creation}, \emph{annihilation} and
\emph{preservation} integrals, respectively.
\end{notn}

The following proposition, which is readily verified, connects the
definition of quantum stochastic integrals of Theorem~\ref{thm: 3.9}
with the classical It\^o integration of simple processes.

\begin{propn}
\label{propn: 3.15}
Suppose the noise dimension space $\Kil$ is finite dimensional, with
orthonormal basis $(e_i)_{i \in \I}$. Let
$F = \smallmat{K}{L}{M}{N}$ be a simple integrand process, let
$t > 0$, and suppose
the partition~$\{ 0 = t_0 < t_1 < \cdots < t_n = t \}$ contains the
points of discontinuities of $F$ on $[ 0, t )$. Then
\[
A^+(L)_t = \sum_{i \in \I} \int_0^t L^i(s) \, \rd A^+( s e_i ) := %
\sum_{i \in \I} \sum_{j = 0}^{n - 1} L^i( t_j ) %
\big( I_\init \ot \aplus ( e_i \indf{[ t_j, t_{j + 1} )} ) \big) %
\quad \text{on } \init \otalg \Exp
\]
and
\[
A^-(M)_t = \sum_{i \in \I} \int_0^t M_i(s) \, \rd A^-( s e_i ) := %
\sum_{i \in \I} \sum_{j = 0}^{n - 1} M_i( t_j ) %
\big( I_\init \ot \aminus( e_i \indf{[ t_j, t_{j + 1} )} ) \big) %
\quad \text{on } \init \otalg \Exp,
\]
where
$L^i(s) := ( \bra{e_i} \ot I_{\init\ot\Fock} ) L(s)$ and
$M_i(s) :=  M(s) ( \ket{e_i} \ot I_{\init\ot\Fock} )$.
\end{propn}

\begin{rem}
The preservation integral $A^\times(N)$ has a similar expression
(see~\cite{Partha}) and the time integral is given by the
straightforward prescription
\[
A^\circ(K)_t := \sum_{j=0}^{n - 1} K(t_j) ( t_{j + 1} - t_j ).
\]
\end{rem}

The following result is the quantum It\^o product formula, or second
fundamental formula. To state it, we define the \emph{quantum It\^o
projection}
\[
\Delta := \mat{0 }{ 0 }{ 0 }{ I_\Kil } \in B( \Khat),
\]
which is ampliated to $\smallmat{0 }{ 0 }{ 0 }{ I_{\Kil \ot \hil} }$
for appropriate choices of $\hil$ without change of notation.

\begin{thm}\label{thm: 3.17}
Let $F$ and $G$ be integrand processes,
let $X_0, Y_0 \in B(\init) \ot I_\Fock$, and, for all $t \in \Rplus$,
set $X_t = X_0 + \Lambda(F)_t$ and $Y_t = Y_0 + \Lambda(G)_t$.
Then
\begin{align*}
\ip{ X_t \uef }{ Y_t \veg } = \ip{ X_0 \uef }{Y_0 \veg } & + %
\int_0^t \Bigl\{ \Bigl\langle \fhat(s) \ot X_s \uef, %
G_s \bigl( \ghat(s) \ot \veg \bigr) \Bigr\rangle \\
 & \qquad + \Bigl\langle F_s \bigl( \fhat(s) \ot \uef \bigr), %
\ghat(s) \ot Y_s \veg \Bigr\rangle \\
& \qquad + \Bigl\langle F_s \bigl( \fhat(s) \ot \uef \bigr), \Delta %
G_s \bigl( \ghat(s) \ot \veg \bigr) \Bigr\rangle \Bigr\} \, \rd s
\end{align*}
for all
$u$, $v \in \init$, $f$, $g \in \Step$ and $t \in \Rplus$.
\end{thm}
\begin{proof}
See \cite[Theorem~3.15]{Lgreifswald}.
\end{proof}

\begin{defn}
\label{defn: vac}
The map
\[
\ExpectVac: B( \init \ot \Fock ) \to B(\init),
\quad
T \mapsto \ExpectVac[ T ] := %
( I_\init \ot \bra{\Vac} ) \, T \, ( I_\init \ot \ket{\Vac} )
\]
is called the \emph{vacuum expectation}. For all $t \in \Rplus$, let
$\shiftK_t$ be the normal $*$-endomorphism of $B(\Fock)$ such that
\[
\shiftK_t\bigl( \Weyl (g) \bigr) = \Weyl ( S_t g ), \qquad %
\text{where }
( S_t g )(s) := \left\{
\begin{array}{ll}
g( s - t ) & \text{ if } s \ges t, \\
 0 & \text{ if }  s < t.
\end{array}\right.
\]
The family $\shiftK := ( \shiftK_t)_{t\ges 0}$ is called the
\emph{CCR flow} of \emph{index} $\dim \Kil$. We set
\[
\sigma_t := \id_{B(\init)} \uwot \shiftK_t %
\qquad \text{ for all } t \in \Rplus.
\]
\end{defn}

\begin{rem}
The vacuum expectation is normal, unital and completely positive, and
the family $\sigma = (\sigma_t)_{t \ges 0}$, is an $E_0$~semigroup
\cite{Arveson} such that:
\begin{equation}
\label{eqn: Esigma} 
\ExpectVac \circ \sigma_t = \ExpectVac \qquad %
\text{for all } t \in \Rplus.
\end{equation}
\end{rem}

\begin{defn}
\label{defn: cocycle}
An $\init$ process $Y$ is a \emph{quantum stochastic cocycle} on
$\init$ if
\[
Y_0 = I_{\init \ot \Fock} \quad \text{ and } \quad
Y_{r+t} = \sigma_r( Y_t) Y_r \qquad \text{for all } r,t \in \Rplus,
\]
and an \emph{elementary} QS cocycle if its vacuum expectation semigroup
$( \ExpectVac[ Y_t ] )_{t \ges 0}$ is norm continuous.
A \emph{Hudson--Parthasarathy cocycle}, or \emph{HP cocycle} in
short, is a unitary elementary QS cocycle
\end{defn}

\begin{rem}
The fact that $ ( \ExpectVac[ Y_t ] )_{t \ges 0}$ is a one-parameter semigroup
follows from the adaptedness relations
\[
\sigma_r( Y_t ) \in B( \init ) \ot %
I_{\Fock_{r)}} \uwot B(\Fock_{[r}) \quad \text{and} \quad %
Y_r \in B(\init \ot \Fock_{r)}) \ot I_{\Fock_{[r}}
\]
and the identity~(\ref{eqn: Esigma}): note that
\[
\ExpectVac[Y_{r+t}] =
\ExpectVac[\sigma_r( Y_{t} )] \, \ExpectVac[Y_{r}] =
\ExpectVac[Y_{t}] \, \ExpectVac[Y_{r}]
\ \text{ for all } r,t \in \Rplus.
\]
\end{rem}

\begin{notn}
Let
\[
B(\Khat \ot \hil)_0 := %
\Bigl\{ T = \smallQSintegrand{T_0^0}{T_0^1}{T_1^0}{T_1^1} \in %
B(\Khat \ot \hil) : T^1_1 = 0 \Bigr\}.
\]
\end{notn}

\begin{thm}\label{thm: 3.19}
\mbox{\par}
\begin{alist}
\item
Let $F \in B( \Khat \ot \init )$. The following are equivalent.
\begin{rlist}
\item
$F = \smallmat{ K }{ L }{ -L^* W }{ W - I_{\Kil \ot \init} }$
where $K = \ri H - \half L^*L$,
for a self-adjoint operator $H$ and unitary operator $W$.
\item
$F^* + F + F^* \Delta F = 0 = F + F^* + F \Delta F^*$.
\item
There is a unitary $\init$~process $U$ such that
\begin{equation}
\label{eqn: QSDE} 
U_t = I_{\init \ot \Fock} + \Lambda( F \cdot U )_t %
\qquad \text{for all } t \in \Rplus,
\end{equation}
where $(F \cdot U)_s := ( F \ot I_\Fock )( I_\Khat \ot U_s )$
for all $s \in \Rplus$.
\end{rlist}
In this case, $U$ is the unique unitary $\init$~process
satisfying~\tu{(\ref{eqn: QSDE})}.
\item Let $U$ be a unitary $\init$~process. The following are
equivalent.
\begin{rlist}
\item $U$ satisfies~\tu{(\ref{eqn: QSDE})} for some operator
$F \in B(\Khat \ot \init)$.
\item $U$ is an HP~cocycle.
\end{rlist}
In this case,
\begin{equation}
\label{eqn: UtoF} 
\ip{ \xhat \ot u }{ ( F + \Delta ) \yhat \ot v } = %
\lim_{t \to 0+} t^{-1} \bigl\langle u \ve( x \indf{[0,t)} ),
( U_t - I_{\init \ot \Fock} ) v \ve( y \indf{[0,t)} ) \bigr\rangle
\end{equation}
for all $u$, $v \in \init$ and $x$, $y \in \Kil$. In particular,
the vacuum expectation semigroup of $U$ has generator $K$.
\item
If $F \in B(\Khat \ot \init )_0$
then \tu{(i)} and \tu{(ii)} of \tu{(a)} have the respective equivalents.
\begin{rlist}
\item $F = \smallmat{K}{L}{-L^*}{0}$,  where
$K + \half L^* L$ is skew-adjoint.
\item $F^* + F + F^* \Delta F = 0$.
\end{rlist}
\end{alist}
\end{thm}
\begin{proof}
Part (a) is covered by Theorems~7.1 and~7.5 of~\cite{LWptrf}. For (b),
see~\cite{Lgreifswald}. The identity~(\ref{eqn: UtoF}) is a
straightforward consequence of~(\ref{eqn: QSDE}), the first
fundamental formula (\ref{eqn: FFF}) and the strong continuity of~$U$.
\end{proof}

\begin{rem}
The quantum stochastic equation (\ref{eqn: QSDE}) is referred to as
the \emph{quantum Langevin equation} in the physics literature
\cite{GarCo, ZolGa}.
\end{rem}

\begin{defn}
\label{defn: 2.10}
Given an HP~cocycle $U$, the unique operator $F$, or triple $(H,L,W)$,
associated with $U$ via~(\ref{eqn: UtoF}) is called its
\emph{stochastic generator}. Conversely, for an operator
$F \in B(\Khat \ot \init)$ having the block-matrix
form given in Theorem~\ref{thm: 3.19}(a)(i), the unique HP~cocycle
satisfying (\ref{eqn: QSDE}) is denoted~$\YF$ or~$U^{(H,L,W)}$.
\end{defn}

\begin{rem}
If $F$ is the stochastic generator of an HP~cocycle then
Theorem~\ref{thm: 3.19} implies that $F^*$ is also such a generator,
since
\[
\begin{bmatrix}
\ri H - \half L^*L & -L^*W \\[0.5ex] L & W - I_{\Kil \ot \init}
\end{bmatrix}^*
=
\begin{bmatrix}
\ri \wt{H} - \half \wt{L}^*\wt{L} & -\wt{L}^*\wt{W} \\[0.5ex]
\wt{L} & \wt{W} - I_{\Kil \ot \init},
\end{bmatrix}
\]
where $\wt{W} = W^*$, $\wt{L} = -W^*L$ and $\wt{H} = -H$.  However, it
is usually not the case that $\YFstar$ and~$( \YF )^*$ are equal. An
exception is when $\init = \Comp$, described in
Example~\ref{examp: pure noise}.
\end{rem}

In this article, we are mainly concerned with the following subclass
of HP~cocycles discussed in~\cite{LQST2}.

\begin{defn}
\label{defn: gauss}
An HP~cocycle is \emph{Gaussian} if its stochastic generator lies in
$B(\Khat \ot \hil)_0$.  Equivalently, its parameterisation has the
form $(H, L, I_{\Kil\ot\init} )$.
\end{defn}

\begin{cor}
\label{cor: bij}
The prescription $(H,L,W) \mapsto U^{(H,L,W)}$ defines a bijection
\[
B(\init)_{\sa} \times B(\init; \Kil \ot \init) \times U(\Kil \ot \init)
\to
\big\{\!
\text{HP~cocycles on $\init$ with noise dimension space $\Kil$}
\big\},
\]
and the restriction $(H,L) \mapsto U^{(H,L,I)}$ defines a bijection
\[
B(\init)_{\sa} \times B(\init; \Kil \ot \init)
\to
\big\{\!
\text{Gaussian HP~cocycles on $\init$ with noise dimension space $\Kil$}
\big\}.
\]
\end{cor}

\begin{example}
\label{examp: pure noise}
[Pure-noise cocycles] For any $z \in \Kil$, setting
$W^z := ( W( z 1_{[0,t)} ) )_{t \ges 0}$ defines an HP~cocycle
on~$\Comp$.  An operator $F \in B( \Khat)$ is the generator of an
HP~cocycle on~$\Comp$ if and only if
\[
F=\begin{bmatrix}
\ri \alpha - \half \norm{z}^2 & - \bra{z} w \\
\ket{z} & w - I_{\Kil}
\end{bmatrix}
\qquad \text{for some } \alpha \in \Real, z \in \Kil \text{ and }
w \in U(\Kil).
\]
The Gaussian pure-noise cocycles are precisely those of the form
$( e^{ \ri \alpha t} W^z_t )_{t\ges 0}$
for some $\alpha \in \Real$ and $z \in \Kil$.

As $B( \Fock_{r)} ) \otimes I_{[r}$ and
$\shiftK_r\bigl( B( \Fock ) \bigr) = I_{r)} \otimes B( \Fock_{[r} )$
commute for all $r \in \Rplus$, the adjoint process~$(\YF)^*$
is equal to the HP~cocycle $\YFstar$ in this case.
\end{example}

\begin{lemma}
\label{lemma: Uprime}
Let $U$ be an HP~cocycle on $\init$ and
let $u$ be a pure-noise HP~cocycle with the same noise dimension space.
Then
\[
\Utilde := ( ( I_\init \ot u_t ) U_t )_{t \ges 0}
\]
is an HP~cocycle on $\init$. Moreover, the stochastic
generators
$\Ftilde \sim ( \Htilde, \Ltilde, \Wtilde )$ of $\Utilde$,
$F \sim (H, L, W)$ of $U$ and
$f \sim (\alpha, \ket{z}, w)$ of $u$ are
related as follows:
\[
\Ftilde = ( f \ot I_\init ) + F + ( f \ot I_\init ) \Delta F
\]
or, equivalently,
\begin{align}
\label{eqn: Wprime} 
\begin{split}
& \Wtilde = ( w \ot I_\init ) W, \\
& \Ltilde = ( w \ot I_\init ) L + \ket{z} \ot I_\init \\
\text{ and } \quad & \Htilde = H +
\tfrac{\ri}{2} \bigl( ( \bra{ w^* z } \ot I_\init ) L - %
L^* ( \ket{ w^* z } \ot I_\init ) \bigr) + \alpha I_\init.
\end{split}
\end{align}
\end{lemma}

\begin{proof}
That the unitary process $\Utilde$ is a QS cocycle follows from the
fact that $\sigma_r(U_t)$ and $I_\init \ot u_r$ commute for all $r$,
$t \in \Rplus$. The quantum It\^o product formula,
Theorem~\ref{thm: 3.17}, implies that
$\Utilde_t = I_{\init \ot \Fock} + \Lambda( \Ftilde \cdot \Utilde )_t$
for all $t \in \Rplus$, where
$\Ftilde = ( f \ot I_\init ) + F + ( f \ot I_\init ) \Delta F$.
It now follows from the uniqueness part of Theorem~\ref{thm: 3.19}
that $\Utilde$ equals the HP~cocycle $\YFtilde$, so that
$\Utilde = U^{( \Htilde, \Ltilde, \Wtilde )}$
where $( \Htilde, \Ltilde, \Wtilde )$ is given by~(\ref{eqn: Wprime}).
\end{proof}

\begin{rem}
More general conditions under which the product of two
elementary QS cocycles is a QS cocycle
are given in \cite{Wills}.
\end{rem}

\begin{defn}
\label{defn: QDS}
A \emph{quantum dynamical semigroup}
$\mathcal{P} = ( \mathcal{P}_t )_{t \ges 0}$ is a semigroup of
completely positive contractive normal maps on $B(\init)$
which is pointwise weak operator continuous.
If $\mathcal{P}_t$ is unital for all $t \in \Rplus$ then $\mathcal{P}$
is called \emph{conservative}.
\end{defn}

\begin{rem}
The generator $\Lindbladian$ of a norm-continuous conservative quantum
dynamical semigroup is expressible in Lindblad form \cite{Lindblad}:
there exists a separable Hilbert space $\Kil$, a self-adjoint
operator~$H \in B(\init)$ and an operator
$L \in B( \init; \Kil \ot \init )$ such that
\begin{equation}
\label{eqn: Lindblad} 
\Lindbladian( a) = %
- \ri [ H, a ] - \half \{ L^*L, a \} + L^* ( I_\Kil \ot a ) L %
\qquad \text{for all } a \in B( \init ),
\end{equation}
where $[\ ,\ ]$ and $\{\ ,\ \}$ denote the commutator and
anti-commutator, respectively.
\end{rem}

\begin{thm}
\label{thm: innerj}
Let $U$ be an HP~cocycle with stochastic generator $(H,L,W)$. For all
$t \in \Rplus$, let
\[
j_t: B(\init) \to B( \init \ot \Fock ); \ %
a \mapsto U^*_t ( a \ot I_\Fock ) U_t,
\]
and let
\begin{equation}\label{eqn: theta}
\theta : B(\init) \to B( \Khat \ot \init ); \ %
a \mapsto
\mat{-\ri [ H, a ] - \half \{ L^*L, a \} + L^* ( I_\Kil \ot a ) L}%
{W^*( ( I_\Kil \ot a ) L - L a )}%
{( L^* ( I_\Kil \ot a ) - a L^* ) W}%
{W^* ( I_\Kil \ot a ) W - I_\Kil \ot a}.
\end{equation}
\begin{alist}
\item If $\jKil := ( \id_{B(\Khat)} \uwot j_t)_{t \ges 0}$,
so that
\[
\jKil_t(A) = %
( I_\Khat \ot U_t )^* ( A \ot I_\Fock ) ( I_\Khat \ot U_t )
\qquad \text{for all } t \in \Rplus \text{ and } %
A \in B(\Khat \ot \init),
\]
then $\bigl( ( \jKil_t \circ \theta )(a) \bigr)_{t \ges 0}$ is an
integrand process for all $a \in B( \init )$ and
\begin{equation}
\label{eqn: mQSDE} 
j_t(a) = a \ot I_\Fock + %
\Lambda\bigl( ( \jKil \circ \theta )(a) \bigr)_t \qquad
\text{for all } a \in B(\init) \text{ and } t \in \Rplus.
\end{equation}
Furthermore, the family $j = ( j_t )_{t \ges 0}$ is the unique
mapping process consisting of normal $*$-homomorphisms that
satisfies~(\ref{eqn: mQSDE}).
\item
The mapping process $j$ obeys the cocycle relation
\[
j_{r+t} = \jhat_r \circ \sigma_r \circ j_t
\qquad \text{for all } r,t \in \Rplus,
\]
where $\jhat_r$ is the normal *-homomorphism from $\Ran \sigma_r$ to
$B(\init \ot \Fock)$ such that
\[
\jhat_r(a \ot b) = j_r(a) ( I_\init \ot b )
\qquad \text{for all } a \in B(\init) \text{ and } %
b \in \Ran \shiftK_r \subseteq B(\Fock).
\]
Moreover, setting
$\mathcal{P} := ( \ExpectVac \circ j_t )_{t \ges 0}$
defines a norm-continuous conservative quantum dynamical semigroup
on~$B(\init)$, the \emph{vacuum expectation semigroup} of~$j$.
\item
For all $a \in B(\init)$, $u$, $v \in \init$ and $x$, $y\in \Kil$,
\[
\bigl\langle \xhat \ot u, %
( \theta(a) + \Delta \ot a )\yhat \ot v \bigr\rangle = %
\lim_{t \to 0+} t^{-1} \bigl\langle u \ve( x \indf{[0,t)} ),
( j_t(a) - a \ot I_{\Fock} ) v \ve( y \indf{[0,t)} ) \bigr\rangle.
\]
In particular, the vacuum expectation semigroup of $j$ has
generator~$\Lindbladian$, given by~(\ref{eqn: Lindblad}).
\end{alist}
\end{thm}
\begin{proof}
That $j$ satisfies~(\ref{eqn: mQSDE}) follows from the quantum It\^o
product formula. In turn, part (c) follows from~(\ref{eqn: mQSDE}),
the first fundamental formula, Theorem~\ref{thm: 3.9}, and the strong
continuity of~$U$. For (b) and the uniqueness part of~(a),
see~\cite{Lgreifswald} and~\cite{LWptrf}.
\end{proof}

\begin{defn}\label{defn: EH}
An \emph{inner Evans--Hudson flow} on $B(\init)$, or
\emph{inner EH~flow} in short, is a mapping process $j$ induced by an
HP~cocycle on $\init$, as above \cite{EvansM}. The map $\theta$ is
called the \emph{stochastic generator} of~$j$.
\end{defn}

\begin{rem}
Let $j$ be an inner EH~flow on $B(\init)$. Using the ampliations
introduced in Theorem~\ref{thm: innerj}, the prescription
$J := ( \jhat_t \circ \sigma_t )_{t\ges 0}$ produces an
$E_0$~semigroup on $B(\init \ot \Fock)$ such that
\[
J_t(A) := U^*_t \sigma_t(A) U_t \qquad \text{for all } %
A \in B(\init \ot \Fock) \text{ and } t \in \Rplus,
\]
where $U$ is any HP~cocycle inducing $j$. In turn, we can recover $j$
from $J$, since $j_t = J_t \circ \iota_\Fock$ for all $t \in \Rplus$,
where the ampliation
\[
\iota_\Fock: B(\init) \to B(\init \ot \Fock); \ %
a \mapsto a \ot I_\Fock.
\]
\end{rem}

Given an HP~cocycle $U$, Lemma~\ref{lemma: Uprime} provides sufficient
conditions for an HP~cocycle $U'$ to induce the same EH~flow as~$U$.
In the next result we show that these conditions are also necessary.

\begin{propn}\label{propn: jprime}
Suppose $j$ and $j'$ are inner EH~flows on $B(\init)$ with noise
dimension space~$\Kil$, induced by HP~cocycles $U$ and $U'$ and having
stochastic generators $(H,L,W)$ and $(H',L',W')$, respectively.  The
following are equivalent.
\begin{rlist}
\item The flows $j$ and $j'$ are equal.
\item The process $( U'_t U^*_t)_{t \ges 0}$ is the ampliation
 to~$\init$ of a pure-noise HP~cocycle.
\item There is a scalar $\alpha \in \Real$, a vector $z \in \Kil$ and
an operator $w \in U(\Kil)$ such that
\begin{align*}
w \ot I_\init & = W'W^*, \\
\ket{z} \ot I_\init  & = L' - ( w \ot I_\init ) L \\
\text{and} \hspace{4em} \alpha I_\init & = H' - H - \tfrac{\ri}{2} %
\bigl( ( \bra{ w^* z } \ot I_\init ) L - %
L^* ( \ket{ w^* z } \ot I_\init ) \bigr).
\end{align*}
\end{rlist}
\end{propn}
\begin{proof}
If (ii) holds then Lemma~\ref{lemma: Uprime} implies that (iii) holds.

If (iii) holds then it is easily verified that $\theta'$, defined from
$(H',L',W')$ rather than $(H,L,W)$, coincides with $\theta$.  Thus (i)
holds by the uniqueness part of Theorem~\ref{thm: innerj}(a).

Finally, suppose that (i) holds, and let $X$ denote the unitary
process $( U'_t U^*_t )_{t \ges 0}$.  For all $t \in \Rplus$, the
operator $X_t$ commutes with all operators in
$B( \init ) \ot I_\Fock$, so $X_t = I_\init \ot u_t$ for some unitary
operator $u_t \in B(\Fock)$.  This implies that $X_r$ commutes with
$\sigma_r( U_t^* )$ for all $r$, $t \in \Rplus$, and so
\begin{align*}
\sigma_r(X_t)X_r = \sigma_r(U'_t) X_r  \sigma_r(U^*_t) = %
\sigma_r(U'_t) U'_r  U^*_r \sigma_r(U_t)^* = U'_{r+t}  U^*_{r+t} = %
X_{r+t}.
\end{align*}
Hence $u = (u_t)_{t \ges 0}$ is a unitary QS cocycle on $\Comp$.
Since $( U')^*$ and $U^*$ are both strongly continuous and unitary,
$u$ is strongly continuous and therefore its vacuum expectation
semigroup $P$ is too. As $P$ is a semigroup on $\Comp$, this implies
that $P$ is norm continuous. Thus $u$ is an HP~cocycle and therefore
(ii) holds.
\end{proof}

\begin{rems}
Given a norm-continuous conservative quantum dynamical semigroup
$\mathcal{P}$ on~$B(\init)$, its generator $\Lindbladian$ is
expressible in Lindblad form~(\ref{eqn: Lindblad}) for some
separable Hilbert space $\Kil$ and operators $H=H^* \in B(\init)$ and
$L \in B(\init; \Kil \ot \init)$.  In turn,
Theorem~\ref{thm:  innerj} implies that the inner EH~flow~$j$ induced
by the HP~cocycle with generator $(H,L,I_{\Kil\ot\init})$ has vacuum
expectation semigroup $\mathcal{P}$.  In this sense, the flow $j$ is a
\emph{stochastic   dilation of $\mathcal{P}$}.

The non-uniqueness of triples $(\Kil,H,L)$ determining the generator
$\Lindbladian$ of a norm-continuous quantum dynamical semigroup on
$B(\init)$ is analysed in~\cite{PaS}; this may be compared to the
non-uniqueness of triples $(H,L,W)$ determining the stochastic
generator $\theta$ of a given inner EH~flow~$j$ characterised in
Proposition~\ref{propn: jprime}.

The construction of stochastic dilations was a major motivation for
the original development of quantum stochastic calculus
\cite{HuP,Partha}.
\end{rems}

We end this summary of standard quantum stochastic calculus by
connecting it to Bhat's analysis of dilations of the above form, in
particular the question of minimality.

\begin{thm}[{\cite[Theorem 9.1]{Bhat}}]\label{thm: 3.23}
Let $j$ be an inner EH~flow. The following are equivalent.
\begin{rlist}
\item
As a stochastic dilation of its vacuum expectation semigroup, the flow
$j$ is \emph{minimal}:
\[
\Linbar \big\{ j_{t_1} ( a_1 ) \cdots j_{t_n} ( a_n ) u \Vac : %
u \in \init, n \ges 1, a_i, \ldots , a_n \in B(\init), %
t_1, \ldots , t_n \in \Rplus \big\} = \init \ot \Fock.
\]
\item The stochastic generator $(H,L,W)$ of any HP~cocycle which
induces~$j$ satisfies
\[
( \bra{z} \ot I_\init ) L \notin \Comp I_\init
\ \text{ for all } z \in \Kil \setminus \{ 0 \}.
\]
\end{rlist}
\end{thm}

\begin{rems}
To see directly that (ii) is independent of the choice of
HP~cocycle which induces $j$, note that for two such HP~cocycles with
stochastic generators $(H_1,L_1,W_1)$ and~$(H_2,L_2,W_2)$, it holds
that
\[
\big\{ ( \bra{z} \ot I_\init ) L_2: %
z \in \Kil \setminus \{ 0 \} \big\} + \Comp I_\init = %
\big\{ ( \bra{z} \ot I_\init ) L_1: %
z \in \Kil \setminus \{ 0 \} \big\} + \Comp I_\init,
\]
by Proposition~\ref{propn: jprime}. This also gives the following
further equivalent condition.
\begin{itemize}
\item[(iii)] The stochastic generator $(H,L,W)$ of any HP~cocycle
which induces $j$ is such that the degeneracy space
$\Kil^L$ equals $\{ 0 \}$; for the definition of $\Kil^L$,
see~(\ref{eqn: kX}).
\end{itemize}

Bhat actually deals with the associated $E_0$~semigroup
$J := ( \jhat_t \circ \sigma_t )_{t\ges 0}$ on $B(\init \ot \Fock)$
which, in view of the remark following Definition~\ref{defn: EH}, is
equivalent.
\end{rems}

\section{Quasifree stochastic calculus}\label{sec: QF SC}

In this section we produce a simplified form of the coordinate-free
multidimensional quasifree stochastic calculus [$\text{LM}_{1,2}$]
with respect to a fixed AW~amplitude $\Sigma = \SigmaAB$ for a Hilbert
space~$\noise$, the \emph{quasifree noise dimension space},
whose conjugate Hilbert space we denote by $( \ol{\noise}, k)$.

In contrast to the approach of [$\text{LM}_{1,2}$],
here we focus on that part of the quasifree stochastic calculus that
may be obtained \emph{inside} the standard theory summarised in Section \ref{sec: QSC}.
Thus, whilst being restricted to HP~cocycles so that stochastic generators
are all bounded, the results developed here do not require faithfulness of
the quasifree states employed.

The conjugate Hilbert space of $L^2(\Rplus; \noise)$ is identified
with $L^2( \Rplus; \kbar)$
(conjugation being defined pointwise: $\ol{f}(t) := \ol{f(t)}$), and
the orthogonal sum
$L^2(\Rplus; \noise) \op L^2( \Rplus; \kbar)$ is identified with
$L^2( \Rplus; \noise \op \kbar )$. Note that we are here working
with the Boson Fock space $\Fock$ over
$L^2( \Rplus; \noise \op \kbar )$.

\begin{motivation}
Let
$\smallrow{ \Sigma_0}{ \Sigma_1 } = %
\smallmat{ \Sigma^0_0 }{ \Sigma^1_0 }{ \Sigma^0_1 }{ \Sigma^1_1 }$
be the block-matrix form of the AW~amplitude $\Sigma$,
with~$\Sigma_0 = %
\smallcol{ \Sigma^0_0 }{ \Sigma^1_0} \in B(\noise; \noise \op \kbar )$
and~$\Sigma_1 = %
\smallcol{ \Sigma^0_1 }{ \Sigma^1_1} \in B(\kbar; \noise \op \kbar )$.
Following Proposition~\ref{propn: 3.15} and the
relations~(\ref{eqn: 2.11a}--b) expressing quasifree creation and
annihilation operators $\aplusSigma$ and  $\aminusSigma$ in terms of
Fock creation and annihilation operators, the following requirements
for quasifree stochastic integration become apparent.

Suppose the quasifree noise dimension space $\noise$ is finite dimensional,
with orthonormal basis~$(e_i)_{i\in\I}$, let~$R$ be a simple
($\noise\ot\init$)-$\init$~process, let $t > 0$ and suppose the
partition $\{ 0 = t_0 < \cdots < t_n = t\}$ contains the points of
discontinuity of $R$ on $[ 0, t )$.
(We are using the symbol $R$ here for convenience;
there is no suggestion of any connection with field operators,
for which the symbol was used earlier.)
Set
\[
\AminusSigma(R)_t = I_1(t) + I_2(t),
\]
where
\begin{align*}
I_1(t) & := \sum_{i \in \I} \sum_{j=0}^{n - 1} %
R_i( t_j ) \bigl( I_\init \ot %
\aminusHplusHbar( \Sigma^0_0 e_i \indf{[t_j, t_{j + 1})}, %
\Sigma^1_0 e_i \indf{[t_j, t_{j + 1})} ) \bigr), \\[1ex]
\text{and} \quad I_2(t) & := \sum_{i \in \I} \sum_{j=0}^{n - 1} %
R_i( t_j ) \bigl( I_\init \ot %
\aplusHplusHbar ( \Sigma^0_1 \ol{e_i} \indf{[t_j, t_{j + 1})}, %
\Sigma^1_1 \ol{e_i} \indf{[t_j, t_{j + 1})} ) \bigr),
\end{align*}
with $\Hil$ denoting $L^2( \Rplus; \noise)$. Note that, for any
$u \in \init$, $f$, $g \in \Stepk$ and $x$, $y \in \noise$,
\begin{align*}
\aminusHplusHbar( \Sigma^0_0 e_i \indf{[t_j, t_{j +1})}, %
\Sigma^1_0 e_i \indf{[t_j, t_{j + 1})} ) \ve( f, \gbar ) & = %
\int_{t_j}^{t_{j+1}} \bigl\langle ( \Sigma^0_0 e_i, \Sigma^1_0 e_i ), %
( f(s), \gbar(s) ) \bigr\rangle \ve(f, \gbar) \, \rd s, \\
\ip{ ( \Sigma^0_0 e_i, \Sigma^1_0 e_i ) }{ ( x, \ybar ) } & = %
\Bigl\langle e_i, (\Sigma_0)^* \binom{x}{\ybar} \Bigr\rangle \\[1ex]
\text{ and } \Bigl\langle e_i, (\Sigma_0)^* %
\binom{x}{\ybar} \Bigr\rangle R_i(t_j) \uefgbar & = %
R(t_j) \big( \dyad{e_i}{e_i} \ot I_{\init \ot \Fock} \bigr) %
\Bigl( (\Sigma_0)^* \binom{ x }{ \ybar } \ot \uefgbar \Bigr).
\end{align*}
Thus
\[
I_1(t) \uefgbar = %
\int_0^t R_s \Bigl( (\Sigma_0)^* \binom{ f(s) }{ \gbar(s) } \ot %
\uefgbar \Bigr) \, \rd s \qquad %
\text{for all } u \in \init \text{ and } f, g \in \Step,
\]
and therefore
\[
I_1(t) \supseteq %
A^-\bigl( R \, ( (\Sigma_0)^* \ot I_{\init \ot \Fock} ) \bigr)_t.
\]
Applying this reasoning to $I_2(t)^*$, and exploiting adaptedness to
commute the terms $R_i(t_j)^*$ and
$I_\init \ot %
\aminusHplusHbar( \Sigma^0_1 \ol{e_i} \indf{[t_j, t_{j +1})}, %
\Sigma^1_1 \ol{e_i} \indf{[t_j, t_{j +1})} )$, where $i \in \I$ and
$j = 0, \ldots , n - 1$, yields the relation
\[
I_2(t)^* = A^-\bigl( R^{\pT *} \, ( (\Sigma_1)^* \ot %
I_{\init \ot \Fock} ) \bigr)_t \qquad \text{on } \init \otalg \ExpK,
\]
where $R^\pT$ is the $\init$-$(\kbar\ot\init)$~process such that
\[
( \bra{ \ol{e_i} } \ot I_{\init\ot\Fock} ) R^\pT_s = %
R_s ( \ket{ e_i } \ot I_{\init\ot\Fock} ) \qquad
\text{for all } i \in \I \text{ and } s\in\Rplus;
\]
$R^\pT$ is said to be \emph{partially transpose} to $R$. It follows
that
\[
I_2(t) \supseteq %
A^+\bigl( ( \Sigma_1 \ot I_{\init\ot\Fock} ) R^\pT \bigr)_t,
\]
and therefore
\[
\AminusSigma(R)_t = A^-( R( (\Sigma_0)^* \ot I_\Fock ) )_t + %
A^+( ( \Sigma_1 \ot I_\Fock ) R^\pT )_t.
\]
Moreover, this also shows, for a suitable
$\init$-$(\noise\ot\init)$~process~$Q$, that
\[
\AplusSigma(Q)_t = A^+\bigl( ( \Sigma_0 \ot I_\Fock ) Q \bigr)_t + %
A^-\bigl( Q^\pT ( (\Sigma_1)^* \ot I_\Fock ) \bigr)_t,
\]
where $Q^\pT$ is the $(\kbar\ot\init)$-$\init$~process partially
transpose to $Q$, given by $( Q^{* \pT *}_t )_{t\ges 0}$.

\noindent
Hence
\[
\AplusSigma(Q)_t + \AminusSigma(R)_t = %
A^+ \Bigl( ( \Sigma \ot I_{\init\ot\Fock} ) \smallcol{Q}{ R^\pT } \Bigr)_t + %
A^- \Bigl( \row{R}{ Q^\pT } ( \Sigma^* \ot I_{\init\ot\Fock} ) \Bigr)_t \ \quad %
\text{for all } t \in \Rplus.
\]
\end{motivation}

The preceding discussion shows clearly the need for a partial
transpose operation for infinite-dimensional $\noise$. A comprehensive
theory is developed in [$\text{LM}_{1,2}$]. Here we specialise to our
context of AW~amplitudes, and it is convenient to concentrate on the
composition of the partial transpose and adjoint operations.

First note that, for any $Y \in B(\hilone; \hil \ot \hiltwo)$, the
quantity
\[
\chil(Y) := \sup\Bigl\{ \bigl( %
\mbox{$\sum_{i \in \I}$} \norm{ Y^* ( e_i \ot u ) }^2 \bigr)^{1/2} : %
u \in \hiltwo, \norm{u} = 1 \Bigr\} \in [0, \infty]
\]
is independent of the choice of orthonormal basis $(e_i)_{i \in \I}$
for $\hil$. When it is finite,
\[
\chil(Y) = \sup\bigl\{ \norm{ Y^* ( I_{\hil} \ot \ket{u} ) }_2 : %
u \in \hiltwo, \norm{u} = 1 \bigr\},
\]
where $\norm{ \cdot }_2$ denotes the Hilbert--Schmidt norm. Let
$HS( \hil; \hil' )$ denote the space of Hilbert--Schmidt operators
from $\hil$ to $\hil'$.

\begin{thm}\label{thm: 3.28}
Let $Y \in B(\hilone; \hil \ot \hiltwo)$.
\begin{alist}
\item
The following are equivalent.
\begin{rlist}
\item
There is an operator $Y^\pc \in B( \hiltwo; \hilbar \ot \hilone)$ such
that
\begin{equation}
\label{eqn: 3.3} 
( \bra{ \ybar } \ot I_{\hilone} ) Y^\pc = %
Y^* ( \ket{y} \ot I_{\hiltwo} ) \qquad \text{for all } y \in \hil.
\end{equation}
\item The quantity $\chil(Y)$ is finite.
\end{rlist}
In this case, the operator $Y^\pc$ is unique and
$\chil(Y) = \norm{ Y^\pc }$; furthermore, $\chil( Y^\pc ) = \norm{Y}$
and~$Y^{\pc \pc} = Y$.

\item Suppose that $\chil(Y) < \infty$, and let
\[
X \in B(\hil'; \hil''), \quad X_1 \in B(\hilone'; \hilone), \quad %
Z_2 \in B(\hiltwo; \hiltwo') \quad \text{and} \quad %
Z \in B(\hil).
\]
The following statements hold.
\begin{rlist}
\item $\chil(Y \ot X) < \infty$ and $( Y \ot X )^\pc = Y^\pc \ot X^*$,
so $\chil(Y \ot X) = \chil(Y) \norm{ X }$;
\item $\chil(Y X_1)  < \infty$ and
$(Y X_1)^\pc = ( I_{\hilbar} \ot X_1^* ) Y^\pc$;
\item $\chil( ( I_{\hil} \ot Z_2 ) Y )  < \infty$ and
$\bigl( ( I_{\hil} \ot Z_2 ) Y \bigr)^\pc = Y^\pc Z_2^*$;
\item $\chil( ( Z \ot I_{\hiltwo} ) Y ) < \infty$ and
$\bigl( ( Z \ot I_{\hiltwo} ) Y \bigr)^\pc = %
( \ol{Z} \ot I_{\hiltwo} ) Y^\pc$.
\end{rlist}
\item Suppose that $\chil(Y \ot I_{\hil'} ) < \infty$ for some
non-zero Hilbert space $\hil'$. Then $\chil(Y) < \infty$
and~$Y^\pc \ot I_{\hil'} = ( Y \ot I_{\hil'})^\pc$.
\item Let $T \in HS(\hil_0; \hil)$ and
$A \in B(\hil_1; \hil_2)$. Then
$\chil( T \ot A ) = \norm{T}_2 \, \norm{A} < \infty$.
\end{alist}
\end{thm}
\begin{proof}
Let $(e_i)_{i \in \I}$ be an orthonormal basis for $\hil$ and note the
trivial identity
\begin{equation}\label{eqn: Y star e} 
Y^* ( \ket{ e_i } \ot I_{\hiltwo} ) u = %
Y^* ( I_{\hil} \ot \ket{ u } )  e_i
\qquad \text{ for all } i \in \I \text{ and } u \in \hiltwo.
\end{equation}
For (a), note first that if $\chil(Y)<\infty$ then the prescription
$u \mapsto \sum_{i \in \I} \ol{ e_i } \ot Y^* ( e_i  \ot  u )$
defines an operator $Y^\pc$ from $\hiltwo$ to $\hilbar \ot \hilone$
which is bounded with norm $\chil(Y)$ and such that
\[
( \bra{ \ybar } \ot I_{\hilone} ) Y^\pc u = %
\sum_{i \in \I} \ip{ e_i }{ y } Y^* ( e_i  \ot  u ) = %
Y^* ( y  \ot  u ) \qquad \text{for all } y \in \hil %
\text{ and } u \in \hiltwo,
\]
so that~(\ref{eqn: 3.3}) holds. Conversely, suppose that an operator
$Y^\pc \in B( \hiltwo; \hilbar \ot \hilone )$
satisfies~(\ref{eqn: 3.3}). Then~\eqref{eqn: Y star e} implies that
\[
\sum_{i \in \I} \norm{ Y^* ( I_{\hil} \ot \ket{ u } )  e_i }^2 = %
\sum_{i \in \I} %
\norm{ ( \bra{ \ol{ e_i} } \ot I_{\hilone} ) Y^\pc u }^2 = %
\norm{ Y^\pc u }^2 \qquad \text{for all } u \in \hiltwo,
\]
so (ii) holds. Uniqueness of the operator $Y^\pc$ is immediate, and
the fact that now $\chilbar( Y^\pc ) = \norm{Y}$ and~$Y^{\pc \pc} = Y$
follows from taking the adjoint of identity~(\ref{eqn: 3.3}).

Parts (b) and (d) are readily verified, and part (c) follows from the
identity
\[
Y^* ( I_{\hil} \ot \ket{ u } ) = %
( I_{\hilone} \ot \bra{ u' } ) %
( Y \ot I_{ \hil' } )^* ( I_{\hil} \ot \ket{ u \ot u' } ),
\]
which is valid for all $u \in \hil$ and any unit vector $u'\in \hil'$.
\end{proof}

\begin{defn}
\label{defn: 3.30}
We let
\[
\Bchil(\hil_1; \hil \ot \hil_2) := %
\bigl\{ Y \in B( \hilone; \hil \ot \hiltwo ): c(Y) < \infty \bigr\},
\]
and note that it is a subspace of $B(\hil_1; \hil \ot \hil_2)$
on which $\chil$ defines a norm. The elements of this space are
$\hil$-\emph{conjugatable} or \emph{partially conjugatable}
operators, and \emph{partial conjugation} is the conjugate-linear
isomorphism
\[
\Bchil( \hilone; \hil \ot \hiltwo ) \to %
\Bchilbar( \hiltwo; \hilbar \ot \hilone ); \ %
Y \mapsto Y^\pc.
\]
An $\hil$-$(\noise\ot\hil')$~process $Q$ is \emph{conjugatable}
if, for all $t \in \Rplus$, the operator $Q_t$ is $\noise$-conjugatable;
in this case $Q^\pc := ( Q^\pc_t )_{t \ges 0}$ is an
$\hil'$-$(\kbar\ot\hil)$~process.
\end{defn}

\begin{rem}
Given any $T \in B( \hilone; \hiltwo )$ and $x \in \hil$,
the operator $\ket { x } \ot T$ is $\hil$-conjugatable,
with the result $( \ket { x } \ot T )^\pc = \ket{ \xbar } \ot T^*$.
In particular, if $\dim \hil < \infty$ then
every operator $Y$ in $B( \hilone; \hil \ot \hiltwo )$ is
$\hil$-conjugatable and $\norm{ Y^\pc } \leq ( \dim \hil ) \norm{Y}$.
\end{rem}

\begin{defn}\label{defn: 3.32}
A $\khat\ot\init$~process $G$ with noise dimension space $\noise$ and
block matrix form $\smallgaussintegrand{K}{Q}{R}$ is a
\emph{$\Sigma$-integrand process on $\init$} if, setting
$\Kil = \noise \op \kbar$,
\begin{alist}
\item the processes $Q$ and $R^*$ are conjugatable,
and
\item the $\Khat \ot \init$~process
$G^\Sigma := \smallgaussintegrand{K}{L}{M}$ is a $\Kil$-integrand
process, in the sense of Definition~\ref{defn: 3.5}, where
\[
L_t := %
( \Sigma \ot I_{\init\ot\Fock} ) \col{ Q_t }{ R^{ * \pc }_t } %
\quad \text{and} \quad M_t := %
\row{ R_t }{ Q^{ \pc * }_t } ( \Sigma^* \ot I_{\init\ot\Fock} ) %
\quad \text{for all } t \in \Rplus.
\]
\end{alist}
In this case, the \emph{quasifree stochastic integral} of $G$ is the
process $\Lambda^\Sigma ( G ) := \Lambda( G^\Sigma )$.
\end{defn}

\begin{rems}
If $G$ is a $\Sigma$-integrand process on $\init$, with block-matrix
form $\smallgaussintegrand{K}{Q}{R}$, then
\begin{equation}\label{eqn: V tilde} 
G^\Sigma = \wh{\Sigma} \, \Gbox \, \wh{\Sigma}^*, \quad \text{where} %
\quad \Gbox :=
\begin{bmatrix}
K & R & Q^{\pc *} \\
Q & 0 & 0 \\
R^{* \pc} & 0 & 0
\end{bmatrix}
\quad \text{and} \quad
\wh{\Sigma} :=
\mat{ 1 }{ 0 }{ 0 }{ \Sigma } \ot I_{\init\ot\Fock}.
\end{equation}

A sufficient condition for a $\khat\ot\init$~process
$\smallgaussintegrand{K}{Q}{R}$ to be a $\Sigma$-integrand process is
that the function
\[
t \mapsto
\norm{ K_t } +
\norm{ Q_t }^2 +
\norm{ Q^\pc_t }^2 +
\norm{ R_t }^2 +
\norm{ R^{* \pc}_t }^2
\]
is locally integrable on $\Rplus$. If $\dim \noise < \infty$ then this
reduces to the local integrability of the function
$t \mapsto \norm{ K_t } + \norm{ Q_t }^2 + \norm{ R_t }^2$.
\end{rems}

We will now show that $\Sigma$-integrability is unaffected by
squeezing. The transformation of integrands resulting from squeezing
the AW~amplitude may be viewed as a change-of-variables formula.

\begin{thm}\label{thm: 4.1}
Let $\Sigmatilde = \Sigma \, M$, where $M$ is
a squeezing matrix for~$\noise$,  and
let $G$ be a $\Sigma$-integrand process. Then there is a
$\Sigmatilde$-integrand process $\Gtilde$ such that
$\Lambda^{\Sigmatilde} ( \Gtilde ) = \Lambda^\Sigma( G )$.
\end{thm}

\begin{proof}
Let $G$ have block-matrix form $\smallQFintegrand{K}{Q}{R}$, let
$M = \MVCP$ as in (\ref{eqn: squeeze}),
and let
\[
\Qtilde_t := ( c V^* \ot I ) Q_t - %
( C s V^* k^{-1} \ot I ) R^{ * \pc }_t
\quad \text{ and } \quad %
\Rtilde_t := R_t ( V c \ot I ) -  Q^{ \pc * }_t ( k V C s \ot I )
\]
for all $t \ges 0$, where $c := \cosh P$, $s:= \sinh P$ and
$I:= I_{\init \ot \Fock}$.
To show that
$\Gtilde := \smallQFintegrand{K}{\Qtilde}{\Rtilde}$ is as desired, it
now suffices to verify the following.
\begin{alist}
\item
The processes $\Qtilde$ and $\Rtilde$ are conjugatable.
\item
For all $t \in \Rplus$, it holds that
\[
( \Sigmatilde \ot I ) \col{ \Qtilde_t }{ \Rtilde^{ * \pc }_t } = %
( \Sigma \ot I ) \col{ Q_t }{ R^{ * \pc }_t } \quad \text{and} \quad %
\row{ R_t }{ Q^{ \pc * }_t } ( \Sigmatilde \ot I )^* = %
\row{ R_t }{ Q^{ \pc * }_t } ( \Sigma \ot I )^*;
\]
equivalently,
\[
( \Sigmatilde \ot I ) %
\mat{ \Qtilde_t }{ \Rtilde^{ * \pc }_t }%
{ \Rtilde^*_t }{  \Qtilde^\pc_t } = %
( \Sigma \ot I ) %
\mat{ Q_t }{ R^{ * \pc }_t }{ R^*_t }{  Q^\pc_t }
\qquad \text{for all } t \in \Rplus.
\]
\end{alist}
Now, Theorem~\ref{thm: 3.28} gives (a), and the following identities:
\begin{align*}
\Rtilde^{ * \pc }_t & = %
( \ol{ c V^* } \ot I ) R^{ * \pc }_t - ( k s C V^* \ot I ) Q_t, \\
\Rtilde^{ * }_t & = %
(  c V^*  \ot I ) R^{ * }_t - ( s C V^* k^{-1} \ot I ) Q^{ \pc }_t \\
\text{and } \quad \Qtilde^{ \pc }_t & = %
( \ol{ c V^* } \ot I ) Q^{ \pc }_t - ( k C s V^* \ot I ) R^*_t
\end{align*}
for all $t \in \Rplus$. Together these imply that
\[
\mat{ \Qtilde_t }{ \Rtilde^{ * \pc }_t }%
{ \Rtilde^*_t }{ \Qtilde^\pc_t } = %
( M \ot I )^{-1} \mat{ Q_t }{ R^{ * \pc }_t }{ R^*_t }{  Q^\pc_t }
\qquad \text{for all } t \in \Rplus,
\]
and so (b) holds as required.
\end{proof}

The following identity is the \emph{first fundamental formula} for quasifree
stochastic integrals.  In view of Theorem~\ref{thm: 3.9}, it holds by
definition.

\begin{propn}\label{propn: 3.35}
Let $G$ be a $\Sigma$-integrand process on $\init$. With the
notation given in \tu{(\ref{eqn: V tilde})},
\[
\big\langle \uef , \Lambda^\Sigma (G)_t \veg \big\rangle = %
\int_0^t \Bigl\langle \wh{ \Sigma^* f}(s) \ot \uef, %
\Gsbox  \, \big( \wh{ \Sigma^* g}(s) \ot \veg \big) %
\Bigr\rangle \, \rd s
\]
for all $u$, $v \in \init$, $f$, $g\in\StepK$ and $t \in \Rplus$.
\end{propn}

The following is readily verified from the definitions. Let
$\Fock^\Hil = \Gamma\bigl( L^2( \Rplus; \Hil ) \bigr)$ for any
choice of~$\Hil$.

\begin{cor}\label{cor: 3.36}
Suppose that the AW amplitude $\Sigma$ is gauge invariant, so has the
form $\SigmaA$, and let $\noise_0 := \Ker A$. Then any
$\Sigma$-integrand process $G$ on $\init$ compresses to a
$\noise_0$-integrand process $G^0$ on $\init$ and $\Lambda(G^0)_t$ is
the compression of $\Lambda^\Sigma( G )_t$ to
$\init \ot \Fock^{\noise_0}$, for all $t \in \Rplus$.
\end{cor}

\begin{rem}
Here $\noise_0$ is being viewed as a subspace of
$\Kil := \noise \op \kbar$ as well as of $\noise$, and
$\Fock^{\noise_0}$ is being identified with the subspace
$\Fock^{\noise_0} \ot \Vac_{\Kil \ominus \noise_0}$ of $\Fock^{\Kil}$.
\end{rem}

This observation shows the quasifree stochastic calculus constructed
here incorporates standard quantum stochastic integrals as well as
purely quasifree stochastic integrals
(\emph{i.e.} QS integrals with respect to integrators which are quasifree
for a \emph{faithful} state), making them useful for the
investigation of repeated interaction systems with particles in a
non-faithful state; see Section~\ref{sec: QRW} and~\cite{B3}.

The following result is the \emph{second fundamental formula} for quasifree
stochastic integrals, and should be compared with
Theorem~\ref{thm: 3.17}. The final term on the right-hand side is the
quasifree It\^o correction term.

\begin{thm}\label{thm: 3.37}
Let $X := ( X_0 + \Lambda^\Sigma( G )_t )_{t\ges 0}$ and
$Y := ( Y_0 + \Lambda^\Sigma( H )_t )_{t\ges 0}$,
where~$G = \smallQFintegrand{K}{Q}{R}$
and~$H = \smallQFintegrand{J}{S}{T}$ are $\Sigma$-integrand processes
and $X_0$, $Y_0 \in B(\init) \ot I_\Fock$. In the notation of
\tu{(\ref{eqn: V tilde})},
\begin{align*}
\ip{ X_t \uef }{ Y_t \veg } & = \ip{ X_0 \uef }{ Y_0 \veg } \\
 & \qquad + %
\int_0^t \Bigl\{ \Bigl\langle \wh{ \Sigma^* f}(s) \ot X_s\uef, %
\Hsbox\bigl( \wh{ \Sigma^* g}(s) \ot \veg \bigr) \Bigr\rangle \\
& \hspace{6em} + \Bigl\langle \Gsbox\bigl( %
\wh{ \Sigma^* f}(s) \ot \uef \big), %
\wh{ \Sigma^* g}(s) \ot Y_s \veg \Bigr\rangle \\
& \hspace{6em} + \Bigl\langle ( \Sigma \ot I_{\init\ot\Fock} ) %
\smallcol{ Q_s }{ R^{* \pc }_s } \uef, %
( \Sigma \ot I_{\init\ot\Fock} ) \smallcol{ S_s }{ T^{* \pc }_s } %
\veg \Bigr\rangle \Bigr\} \, \rd s
\end{align*}
for all $u$, $v \in \init$, $f$, $g\in\StepK$ and $t \in \Rplus$.
\end{thm}

\begin{proof}
This follows immediately from
Theorem~\ref{thm: 3.17},
Definition~\ref{defn: 3.32}
and the identity
\[
\bigl\langle G^\Sigma_s ( \xhat \ot \uef ), %
\Delta H^\Sigma_s ( \yhat \ot \veg ) \bigr\rangle = %
\Bigl\langle ( \Sigma \ot I_{\init\ot\Fock} ) %
\smallcol{ Q_s }{ R^{* \pc }_s } \uef, %
( \Sigma \ot I_{\init\ot\Fock} ) %
\smallcol{ S_s }{ T^{* \pc }_s } \veg \Bigr\rangle,
\]
which holds for all $x$, $y \in \Kil$, $u$, $v \in \init$, $f$,
$g\in\StepK$ and $s \in \Rplus$.
\end{proof}

\begin{thm}\label{thm: 3.39}
Let $G \in B( \khat \ot \init )_0$. The following are equivalent.
\begin{rlist}
\item The operator $G$ has block-matrix form
$\smallQFintegrand{K}{Q}{-Q^*}$, where $Q$ is conjugatable and
\[
K + K^* + L^* L = 0 \qquad \text{for the operator } L := %
( \Sigma \ot I_\init ) \col{ Q }{-Q^{ \pc } }.
\]
\item
There is a unitary $\init$~process $U$ with noise dimension space
$\Kil = \noise \op \kbar$ such that
\begin{alist}
\item[(a)]
$G \cdot U := ( ( G \ot I_\Fock ) ( I_{\khat} \ot U_t ) )_{t \ges 0}$
is a $\Sigma$-integrand process, and
\item[(b)]
$U_t = I_{\init \ot \Fock} + \Lambda^\Sigma ( G \cdot U )_t$
for all $t \in \Rplus$.
\end{alist}
\end{rlist}
If either condition holds then $U$ is the unique $\init$~process
satisfying \tu{(a)} and \tu{(b)} of \tu{(ii)}.
\end{thm}

\begin{proof}
Suppose that (i) holds
and set
\begin{equation*}
F = G^\Sigma := \mat{ I_\init }{ 0 }{ 0 }{ \Sigma \ot I_\init }
\begin{bmatrix}
K & -Q^* & Q^{ \pc * } \\
 Q & 0 & 0 \\
-Q^{ \pc } & 0 & 0
\end{bmatrix}
\mat{ I_\init }{ 0 }{ 0 }{ \Sigma \ot I_\init }^* = %
\mat{K}{L}{-L^*}{0}.
\end{equation*}
Then $F \in B( \Khat \ot \init )_0$
and $F^* + F + F^* \Delta F = 0$.
Appealing to Theorem~\ref{thm: 3.19} and Definition~\ref{defn: 2.10},
there exists a unitary process $U := Y^F$.
Since $( G \cdot U )^\Sigma = F \cdot U$, so $G \cdot U$ is a
$\Sigma$-integrand process and
$\LambdaSigma( G \cdot U )_t = \Lambda( F \cdot U )_t = %
U_t - I_{\init \ot \Fock}$
for all $t \in \Rplus$, hence (ii) holds.

Conversely, suppose that (ii) holds for a unitary $\init$~process $U$,
and let $\smallgaussintegrand{K}{Q}{R}$ be the block-matrix form of
$G$. Theorem~\ref{thm: 3.28} implies that the operators
$Q$ and $R^*$ are conjugatable, and
\begin{equation}\label{eqn: WUF} 
( G \cdot U )^\Sigma = F \cdot U, \qquad \text{where } %
F = G^\Sigma := \mat{ I_\init }{ 0 }{ 0 }{ \Sigma \ot I_\init } %
\begin{bmatrix}
K & R & Q^{ \pc * } \\
Q & 0 & 0  \\
R^{ * \pc } & 0 & 0
\end{bmatrix}
\mat{ I_\init }{ 0 }{ 0 }{ \Sigma \ot I_\init }^*.
\end{equation}
Assumption (b) gives that
$U_t = I_{\init \ot \Fock} + \Lambda( F \cdot U )_t$ for all
$t \in \Rplus$, and so, by Theorem~\ref{thm: 3.19}, it holds that
$F^* + F + F^* \Delta F = 0$ and $U = Y^F$. In particular, the
uniqueness claim is established. The condition
$F^* + F + F^* \Delta F = 0$ is equivalent to
\begin{alist}
\item
$\row{ R }{ Q^{ \pc * } } ( \Sigma^* \ot I_{\init} ) = %
-\Bigl( ( \Sigma \ot I_{\init} ) %
\smallcol{ Q }{ R^{ * \pc } } \Bigr)^*$ and
\item
$0 = K^* + K +  L^* L$,
where $L = ( \Sigma \ot I_{\init} ) \smallcol{ Q }{ R^{ * \pc } }$,
\end{alist}
so it remains to prove that $X := Q + R^* = 0$. Note that (a) is
equivalent to
$( \Sigma \ot I_{\init} ) \smallcol{ X}{ X^{ \pc } } = 0$ and, in
terms of the parameterisation $\SigmaAVCP$ of the AW~amplitude
$\Sigma$ given in~(\ref{eqn: SUCP}) and the notation $k$ for the
conjugation map from $\noise$ to $\ol{\noise}$,
this is equivalent
to
\begin{equation}
\label{eqn: XXc} 
\mat{ \cosh A \cdot V \cosh P \ot I_\init }%
{k \sinh A \cdot V C \sinh P \ot I_\init }%
{\cosh A \cdot V \sinh P \cdot C k^{-1} \ot I_\init }%
{\ol{ \sinh A \cdot V \cosh P } \ot I_\init } %
\col{ X}{ X^{  \pc }} = 0.
\end{equation}
It follows from (\ref{eqn: XXc}) that
$X = - ( \tanh P \cdot C k^{-1} \ot I_\init ) X^{ \pc}$,
and so,
by Theorem~\ref{thm: 3.28} and the fact that $C$ commutes with $P$
and $C^2 = I_\noise$,
\begin{align*}
X & = ( \tanh P \cdot C k^{-1} \ot I_\init ) %
\bigl( ( \tanh P \cdot C k^{-1} \ot I_\init ) %
X^{ \pc } \bigr)^{\pc} \\
 & = ( \tanh P \cdot C k^{-1} \ot I_\init ) %
( k \tanh P \cdot C \ot I_\init ) X \\
 & = ( \tanh^2 P  \ot I_\init ) X,
\end{align*}
thus
$0 = \bigl( ( I_\noise - \tanh^2 P )  \ot I_\init \big) X = %
( \cosh^2 P   \ot I_\init )^{-1} X$ and so $X=0$.
\end{proof}

\begin{rem}
From the preceding proof, we see that the unique unitary
$\init$~process $U$ determined by an operator
$G \in B(\khat \ot \init )_0$ satisfying Theorem~\ref{thm: 3.39}(i)
equals $\YF$, where $F = G^\Sigma$ as defined in~(\ref{eqn: WUF}). In
particular, $U$ is an HP~cocycle.  Cocycle aspects of quasifree
processes are further investigated in~\cite{LM2}.
\end{rem}

\begin{defn}
\label{defn: QFcocycle}
An HP~cocycle $U$ on $\init$ with noise dimension space
$\noise \op \kbar$ is \emph{$\Sigma$-quasifree}
and has~\emph{$\Sigma$-generator $G$} if
$U = \YF$ for $F = G^\Sigma$,
in which
$G \in B( \khat \ot \init )_0$ has the block-matrix
form~$\smallQFintegrand{K}{Q}{-Q^*}$, where $Q$ is $\noise$-conjugatable
and
\begin{equation}\label{eqn: qf struct} 
K + K^* + L^* L = 0 \qquad \text{for the operator }
L := ( \Sigma \ot I_\init ) \col{ Q }{-Q^{ \pc } }.
\end{equation}
\end{defn}

\begin{rem}
Thus $\Sigma$-quasifree HP~cocycles form a subclass of the collection
of Gaussian HP~cocycles with noise dimension space $\Kil$ having a
decomposition $\noise \op \kbar$.
\end{rem}

\begin{example}
\label{examp: qf pure noise}
[Pure-noise cocycles]
For a gauge-invariant AW~amplitude $\Sigma = \SigmaA$, the
quasifree pure-noise cocycles are of the form
$\big( e^{\ri \alpha t} \WeylSigma( x \indf{[0,t)} ) \big)_{t \ges 0}$
for some $x \in \noise$ and $\alpha \in \Real$, with corresponding
$\Sigma$-generator
$\left[\begin{smallmatrix}
\ri \alpha - \half \norm{ \sqrt{ \cosh 2A }\, x}^2 & & -\bra{x} \\
 \ket{ x }  & &  0
\end{smallmatrix}\right]$.
\end{example}

\begin{cor}
\label{cor: 3.11}
Let $U$ be a Gaussian HP~cocycle on $\init$ with noise dimension space
$\noise \op \kbar$ and stochastic generator
$\smallgaussintegrand{K}{L}{-L^*}$, let $\smallcol{L_1}{L_2}$ be the
block matrix form of $L$, and suppose that the AW-amplitude
is gauge-invariant, say $\Sigma = \SigmaA$.  Then the following are
equivalent.
\begin{rlist}
\item The cocycle $U$ is a $\Sigma$-quasifree HP~cocycle.
\item The operator $L$ equals $( \Sigma \ot I_\init ) \smallcol{Q}{-Q^\pc}$
for a $\noise$-conjugatable operator
$Q \in B(\init; \noise \ot \init)$.
\item The operator $L_1$ is $\noise$-conjugatable and
$L_2 = - ( \ol{ \tanh A } \ot I_\init ) L_1^\pc$.
\item The operator $L_2$ is $\kbar$-conjugatable and
$L_2^\pc = - ( \tanh A  \ot I_\init ) L_1$.
\end{rlist}
When these hold, the cocycle $U$ has $\Sigma$-generator
$\smallmat{K}{Q}{-Q^*}{0}$
and
\begin{equation}\label{eqn: zLxQ} 
\bigl( \bra{ x } \ot I_\init \bigr) Q - %
Q^* \bigl( \ket{ x } \ot I_\init \bigr) = %
\bigl( \bra{ \Sigma \iota(x) } \ot I_\init \bigr) L - %
L^* \bigl( \ket{ \Sigma \iota(x) } \ot I_\init \bigr) \qquad %
\text{for all } x \in \noise.
\end{equation}
\end{cor}
\begin{proof}
By Theorem~\ref{thm: 3.39} and Definition~\ref{defn: QFcocycle},
(i) is equivalent to (ii), and these imply that $U$ has
$\Sigma$-generator $\smallmat{K}{Q}{-Q^*}{0}$. Properties of the
partial conjugation, Theorem~\ref{thm: 3.28}, now imply that (ii) is
equivalent to (iii); they also imply that (iii) is equivalent to (iv).
When these conditions hold, since
\[
\bigl( \bra{ \Sigma \iota(x) } \ot I_\init \bigr) L = %
\bigl( \bra{ x } \cosh^2 A \ot I_\init \bigr) Q + %
Q^* \bigl( \sinh^2 A\ket{ x } \ot I_\init \bigr)
\qquad \text{for all } x \in \noise,
\]
the identity (\ref{eqn: zLxQ}) follows from the fact that
$\cosh^2 A - \sinh^2 A = I_\noise$.
\end{proof}

\begin{thm}\label{thm: innerQF}
Let $U$ be a $\Sigma$-quasifree HP~cocycle with $\Sigma$-generator
$\smallQFintegrand{K}{Q}{-Q^*} \in B( \khat \ot \init)_0$, and let
$j$ be the corresponding inner EH~flow.
Set $L:= ( \Sigma \ot I_{\init} ) \smallcol{Q}{ -Q^\pc }$
and $H := \tfrac{1}{2 \ri} ( K - K^*)$,
and define the map
\[
\psi: B(\init) \to B(\khat \ot \init); \ a \mapsto %
\mat{-\ri [H,a] - \half \{ L^*L, a \} + L^* (I_\noise \ot a ) L}%
{( I_\noise \ot a ) Q - Qa}{Q^*( I_\noise \ot a ) - aQ^*}{0}.
\]
Then $\bigl( (\jnoise_t \circ \psi)(a) \bigr)_{t\ges 0}$ is a
$\Sigma$-integrand process for all $a \in B( \init )$, where
$\jnoise_t := \id_{B(\khat)} \uwot j_t$, and
\[
j_t(a) = a \ot I_\Fock + %
\Lambda^\Sigma \bigl( ( \jnoise \circ \psi )(a) \bigr)_t %
\qquad \text{ for all } a \in B(\init) \text{ and } t \in \Rplus.
\]
\end{thm}

\begin{proof}
It is straightforward to verify that
\[
\big( ( \jnoise \circ \psi )(a) \big)^\Sigma_s = %
( \jKil_s \circ \theta )(a) \qquad %
\text{ for all } a \in B(\init) \text{ and } s \in \Rplus,
\]
where $\jKil_s := \id_{B(\Khat)} \uwot j_s$ for
$\Kil = \noise \op \kbar$, and $\theta$ is the map from $B(\init)$ to
$B( \Khat \ot \init)$ defined in~(\ref{eqn: theta}). It therefore
follows from Theorem~\ref{thm: innerj} that
\[
j_t(a) - a \ot I_\Fock =
\Lambda \big( ( \jKil \circ \theta )(a) \big)_t =
\Lambda^\Sigma \big( ( \jnoise \circ \psi )(a) \big)_t
\ \text{ for all $a \in B(\init)$ and $t\in\Rplus$},
\]
as claimed.
\end{proof}

\section{Uniqueness questions}\label{sec: uniqueness}

In this section, issues of uniqueness are considered. We begin with
the question of uniqueness of AW~amplitudes for quasifree HP~cocycles.
Given an HP~cocycle $U$ with noise dimension space~$\Kil$ and stochastic
generator $F = \smallmat{K}{L}{-L^*W }{W - I_{\Kil \ot \init} }$, we
examine the class of pairs~$( \Sigma, Q )$ such that
\[
\Sigma \text{ is an AW~amplitude}, \quad
Q \text{ is a $\noise$-conjugatable operator and } \quad
( \Sigma \ot I_\init ) \col{ Q }{ -Q^\pc } = L,
\]
so that
$G := \smallQFintegrand{K}{Q}{ -Q^* }$ is a $\Sigma$-quasifree
generator and $F = G^\Sigma$. Immediate necessary conditions for this
class to be non-empty are that the HP~cocycle~$U$ is Gaussian,
thus~$F \in B( \Khat \ot \init )_0$,
so~$W = I_{\Kil \ot \init}$,
and $\Kil$ has a decomposition
$\kil \op \kbar$, so $\Kil$ must not have finite odd dimension.

We also consider the uniqueness of quasifree HP~cocycles implementing
a given EH~flow $j$ and relate this to the minimality of $j$ as a
stochastic dilation of its expectation semigroup.

For the remainder of this section, we fix a quasifree noise dimension
space $\noise$, and set $\Kil = \noise \op \kbar$.
Theorem~\ref{thm: 4.1} has the following consequence.

\begin{cor}\label{cor: SigmaM}
Let $\Sigmatilde = \Sigma \, M$, where $\Sigma$ and $M$ are an
AW~amplitude and squeezing matrix for~$\noise$, respectively. Then
every $\Sigma$-quasifree HP~cocycle is also $\Sigmatilde$-quasifree.
\end{cor}

In light of the above corollary, we restrict to gauge-invariant
AW~amplitudes for the rest of this section.  For an operator
$X \in B(\init; \noise \ot \init)$, let the
\emph{$\noise$-degeneracy space} of $X$ be
\begin{equation}\label{eqn: kX}
\noise^X := %
\bigl\{ x \in \noise : ( \bra{x} \ot I_\init ) X = 0 \bigr\}.
\end{equation}

\begin{propn}\label{cor: 4.3}
Let $\Sigma=\SigmaA$ be a gauge-invariant AW~amplitude for $\noise$,
and suppose $U$ is a~$\Sigma$-quasifree HP~cocycle with stochastic
generator $\smallgaussintegrand{K}{L}{-L^*}$ and $\Sigma$-generator
$\smallQFintegrand{K}{Q}{-Q^*}$, where~$L$ has block-matrix form
$\smallcol{L_1}{L_2}$. Then
\begin{equation}\label{eqn: kLoneQ}
\noise^{L_1} = \{ 0 \} \iff \noise^Q  = \{ 0 \}.
\end{equation}
Furthermore, if $\Sigmatilde = \Sigma_{\Atilde}$ is another
gauge-invariant AW~amplitude for~$\noise$, then the following are
equivalent.
\begin{rlist}
\item
The cocycle $U$ is also $\Sigmatilde$-quasifree.
\item
$\big( ( \tanh \Atilde - \tanh A ) \ot I_\init \big) L_1 = 0$.
\end{rlist}
\end{propn}

\begin{proof}
Corollary~\ref{cor: 3.11} implies that $L_2$ is $\kbar$-conjugatable
and $Q$ is $\noise$-conjugatable, with
\begin{equation}\label{eqn: Ltwoc}
L_1 = ( \cosh A \ot I_\init ) Q \qquad \text{and} \qquad %
L_2^\pc = -( \tanh A \ot I_\init ) L_1.
\end{equation}
Thus~(\ref{eqn: kLoneQ}) follows from the invertibility of~$\cosh A$.
Corollary~\ref{cor: 3.11} also implies that (i) holds if and only if
$L_2^\pc = - ( \tanh \Atilde \ot I_\init ) L_1$. Therefore (i) and
(ii) are equivalent, by~(\ref{eqn: Ltwoc}).
\end{proof}

For an HP~cocycle $U$ with noise dimension space $\noise \op \kbar$,
let
\[
\Xi(U) := \bigl\{ \Sigma \in AW_0(\noise): %
U \text{ is $\Sigma$-quasifree} \bigr\}
\]
be the set of gauge-invariant AW~amplitudes for $\noise$ for which $U$
is $\Sigma$-quasifree.

\begin{cor}
\label{cor: XiAQ}
Let $U$ be an HP~cocycle with stochastic generator
$\smallmat{K}{L}{-L^*}{0}$. If $U$ is quasifree with respect to a
gauge-invariant AW~amplitude $\SigmaA$ then
\begin{align*}
\Xi ( U ) & = \bigl\{ \Sigma_{\Atilde} : %
\Atilde \in B(\noise)_+ \text{ and } %
\Ran ( \tanh \Atilde - \tanh A ) \subseteq \noise^{L_1} \bigr\} \\
& = \bigl\{ \Sigma_{\tanh^{-1} ( X + \tanh A ) } : %
X \in B(\noise)_{\sa}, \ \spec ( X + \tanh A ) \subseteq [0,1) %
\text{ and } \Ran X \subseteq \noise^{L_1} \bigr\}.
\end{align*}
In particular, if $\noise^{L_1} = \{ 0 \}$ then $U$ is quasifree with
respect to at most one gauge-invariant AW~amplitude.
\end{cor}

We now turn to the question of implementability of inner EH~flows by
quasifree HP~cocycles.

\begin{propn}\label{propn: 4.7} 
Let $U$ and $\Utilde$ be quasifree HP~cocycles on $\init$ with respect
to a gauge-invariant AW~amplitude $\Sigma$ for $\noise$, and let
$\smallQFintegrand{K}{Q}{-Q^*}$ and
$\smallQFintegrand{\Ktilde}{\Qtilde}{-\Qtilde^*}$ be their respective
$\Sigma$-generators. The following are equivalent.
\begin{rlist}
\item The cocycles $U$ and $\Utilde$ induce the same inner EH~flow.
\item There exist $x \in \noise$ and $\alpha \in \Real$ such that
\[
\Qtilde - Q = \ket{ x } \ot I_\init \quad \text{and} \quad
\Htilde - H - \tfrac{\ri}{2} %
\bigl( ( \bra{ x } \ot I_\init ) Q - %
Q^* ( \ket{ x } \ot I_\init ) \bigr) = \alpha I_\init,
\]
where $H := \tfrac{1}{2 \ri} ( K - K^*)$ and
$\Htilde := \tfrac{1}{2 \ri} ( \Ktilde - \Ktilde^*)$.
\end{rlist}
\end{propn}

\begin{proof}
Let $C$ and $T$ denote $\cosh A$ and $\tanh A$, respectively, where
$\Sigma = \SigmaA$, and let
\[
L := ( \Sigma \ot I_\init ) \col{ Q }{ -Q^\pc }, %
\quad K := \ri H - \half L^*L, \quad %
\Ltilde := ( \Sigma \ot I_\init ) \col{ \Qtilde }{ -\Qtilde^\pc } %
\quad \text{and} \quad %
\Ktilde := i\Htilde - \half \Ltilde^*\Ltilde.
\]
By Proposition~\ref{propn: jprime}, (i) is equivalent the existence of
$z = (z_1, \ol{z_2} ) \in \noise \op \kbar$ and $\alpha \in \Real$
such that
\begin{equation}\label{eqn: LtildeL}
\Ltilde - L = \ket{ z } \ot I_\init \quad \text{and} \quad %
\Htilde - H - \alpha I_\init = %
\tfrac{\ri}{2 } \bigl( ( \bra{ z } \ot I_\init ) L - %
L^* ( \ket{ z } \ot I_\init ) \bigr).
\end{equation}
If $z = ( z_1, \ol{z_2} ) \in \noise \op \kbar$ and $\alpha \in \Real$
are such that (\ref{eqn: LtildeL}) holds then
\[
0 = ( T \ot I_\init ) %
\bigl( L_1 + \ket{z_1} \ot I_\init - \wt{L_1} \bigr) = %
-L_2^\pc + \ket{T z_1} \ot I_\init + \wt{L_2}^\pc = %
\ket{z_2 + T z_1} \ot I_\init,
\]
so $z_2 = - T z_1$, and therefore
$z = \Sigma \iota(x)$, where $x = C^{-1} z_1$.
It follows from~(\ref{eqn: zLxQ}) that (ii) holds.

Conversely, suppose that (ii) holds, with $x \in \noise$ and
$\alpha \in \Real$, and set $z := \Sigma \iota ( x )$. Then
\[
\Ltilde - L = ( \Sigma \ot I_\init ) %
\col{ \Qtilde - Q }{ Q^\pc -\Qtilde^\pc } = %
( \Sigma \ot I_\init ) %
\col{ \ket{ x } \ot I_\init }{ - \ket{ \xbar } \ot I_\init } = %
\ket{ z } \ot I_\init,
\]
so
$\Qtilde - Q = \ket{ C^{-1} z_1 } \ot I_\init = \ket{ x } \ot I_\init$
and, by~(\ref{eqn: zLxQ}), condition~(\ref{eqn: LtildeL}) is
satisfied.
\end{proof}

\begin{thm}
\label{thm: 4.6}
Let $j$ be an inner EH~flow which is a minimal dilation of its vacuum
expectation semigroup. Then there is at most one gauge-invariant
AW~amplitude $\Sigma$ such that $j$ is induced by a $\Sigma$-quasifree
HP~cocycle.
\end{thm}
\begin{proof}
Suppose that $j$ is induced by a $\Sigma$-quasifree HP~cocycle~$U$ and
a $\Sigmatilde$-quasifree HP~cocycle~$\Utilde$, where $\Sigma=\SigmaA$
and $\Sigmatilde = \Sigma_{\Atilde}$ are gauge-invariant AW~amplitudes
for $\noise$. Then $U$ and $\Utilde$ are Gaussian and so have
stochastic generators of the form $\smallgaussintegrand{K}{L}{-L^*}$
and $\smallgaussintegrand{K}{\Ltilde}{-\Ltilde^*}$ respectively.
Letting $\smallgaussintegrand{K}{Q}{-Q^*}$ and
$\smallgaussintegrand{\Ktilde}{\Qtilde}{-\Qtilde^*}$ be their
respective quasifree generators, it follows that
\[
( \Sigma \ot I_\init ) \col{Q}{-Q^\pc} = L \qquad \text{and} \qquad %
( \Sigmatilde \ot I_\init ) \col{\, \Qtilde}{-\Qtilde^\pc} = \Ltilde,
\]
and Proposition~\ref{propn: jprime} implies that
$\Ltilde = L + \ket{z} \ot I_\init$ for some
$z = (z_1, \ol{z_2})$ in $\noise \op \kbar$.
If~$T := \tanh A$ and $\Ttilde := \tanh  \Atilde$ then
\begin{align*}
\bigl( ( T - \Ttilde ) \ot I_\init \bigr) \wt{L}_1 & = %
( T \ot I_\init ) \bigl( L_1 + \ket{z_1} \ot I_\init \bigr) - %
( \Ttilde  \ot I_\init ) \wt{L_1} \\
 & = -L_2^\pc + \ket{T z_1} \ot I_\init + \wt{L_2}^\pc \\
 & = \ket{z_2 + T z_1} \ot I_\init,
\end{align*}
so if $y \in \Ran ( T - \Ttilde )^*$ then
\[
\bigl( \bra{ (y,0) } \ot I_\init \bigr) \Ltilde = %
\bigl( \bra{ y } \ot I_\init \bigr) \wt{L_1} \in \Comp I_\init.
\]
Therefore, by Theorem~\ref{thm: 3.23}, the minimality of $j$ implies
that $\Ran ( T - \Ttilde )^* = \{ 0 \}$, so $\Ttilde = T$,
$\Atilde = A$ and $\Sigmatilde = \Sigma$.
\end{proof}

\section{Quantum random walks}\label{sec: QRW}

In this section we first review the basic theory of unitary quantum
random walks for particles in a vector state and their convergence to
quantum stochastic cocycles \cite{B1}; for an elementary treatment via
the semigroup decomposition of quantum stochastic cocycles,
see~\cite{BGL}. Stronger theorems for more general walks may be found
in~\cite{B2}, for particles in a faithful normal state, and
in~\cite{B3}, for particles in a general normal state. We then
construct quantum random walks in the repeated-interactions model for
particles in a faithful normal state~$\rho$. Thus let $\rho$ be such a
state on $B(\particle)$, for a Hilbert space $\particle$. Under the
assumption that the interaction Hamiltonian $\Hint$ has no diagonal
component with respect to the eigenspaces of the density matrix
of~$\rho$, we demonstrate convergence to HP~cocycles of the
form~$U \ot I$ where $I$ is the identity operator of the Fock space
over $L^2(\Rplus; \Kil_0)$ for a subspace~$\Kil_0$ of the GNS space
of~$\rho$. The construction yields a quasifree noise dimension space
$\noise$ together with natural conjugate space $\kbar$ and, under the
assumption of exponential decay of the eigenvalues of the density
matrix corresponding to $\rho$, a gauge-invariant AW~amplitude
$\Sigma(\rho)$ for $\noise$. We then show that~$U$ is
$\Sigma(\rho)$-quasifree, assuming only that~$\Hint$ is
$\particle$-conjugatable.  We also show that if the lower-triangular
matrix components of $\Hint$ are strongly linearly independent then
$\Sigma(\rho)$ is the unique gauge-invariant AW~amplitude with respect
to which~$U$ is quasifree.

\subsection*{Particles in a vector state}
For this subsection, we fix a noise dimension space $\Kil$.

\begin{defn}
The \emph{toy Fock space} $\Upsilon$ over $\Kil$ is the
tensor product of a sequence of copies of~$\Khat := \Comp \op \Kil$
with respect to the constant stabilising sequence given by
$\vac := %
\left(\begin{smallmatrix} 1 \\[0.5ex] 0 \end{smallmatrix}\right)$:
\[
\Upsilon := \bigotimes_{n=0}^\infty \big( \Khat, \vac \big).
\]
We also set
\[
\Upsilon_{[m} := \bigotimes_{n=m}^\infty \big( \Khat, \vac \big)
\qquad \text{for all } m \ges 1
\]
and denote the identity operator on $\Upsilon_{[m}$ by $I_{[m}$.
\end{defn}

As is readily verified \cite{B1,BGL}, toy Fock space over $\Kil$
approximates Boson Fock space over $\Kil$ in the following sense. Let
$\Fock_J = \Gamma\bigl( L^2( J; \Kil ) \bigr)$ for any subinterval
$J \subseteq \Rplus$, with $\Vac_J$ its vacuum vector, and, for all
$\tau > 0$, let
\[
D_\tau: \Upsilon \to \bigotimes_{n=0}^\infty %
\bigl( \Fock_{ [n\tau, (n+1)\tau) }, %
\Vac_{ [n\tau, (n+1)\tau) } \bigr) = \Fock
\]
be the isometric linear operator such that
\begin{align*}
\left( \binom{ 1 }{ x_n } \right)_{n \ges 0} \mapsto %
\bigotimes_{n=0}^\infty %
\bigl( 1, \tau^{-1/2} x_n \indf{ [n \tau, (n+1)\tau) } %
\bigr)
\end{align*}
for any finitely-supported sequence $( x_n )$ in $\Kil$. Then
$D_\tau D_\tau^* \to I_\Fock$ in the strong operator topology as
$\tau \to 0+$.

\begin{defn}
\label{defn: QRWgen}
For any $G \in U( \Khat \ot \init)$, the \emph{quantum random walk
generated by $G$} is the sequence
$( U_n )_{n\ges 0}$ in $B( \init \ot \Upsilon)$ defined recursively
as follows:
\[
U_0 = I_{\init \ot \Upsilon} \quad \text{and} \quad %
U_{n+1} = ( \sigma_n \circ \iota )( G ) U_n %
\qquad \text{for all } n \ges 0,
\]
where the normal $*$-monomorphism
\[
\iota : B(\Khat \ot \init) \to B( \init \ot \Upsilon ); \ %
A \ot X \mapsto X \ot A \ot I_{[1}
\]
and $\sigma_n := \id_{B(\init)} \uwot \sigma_n^\Upsilon$
is the ampliation of the right shift $*$-endomorphism of $B(\Upsilon)$
with range~$I_{\Khat^{\otimes n}} \ot B(\Upsilon_{[n})$.

Scaling maps on $B( \Khat \ot \init )$  are defined by setting
\[
\tauscale \left( \mat{ A}{B}{C}{D} \right) = %
\mat{ \tau^{-1} A }{  \tau^{-1/2} B }{ \tau^{-1/2} C  }{  D }
\qquad \text{for all } \tau > 0.
\]
\end{defn}

\begin{rems}
If the generator is an elementary tensor $A \ot X$
then the quantum random walk
takes the simple form
\[
\big( X^n \ot A^{\ot n} \ot I_{[n} \big)_{n \ges 0}.
\]
For us here, generators are of the form $\exp{\ri H}$ for operators
$H \in B( \Khat \ot \init)_{\sa}$.

In~\cite{BGL} we worked with \emph{left} QRW's
and generators in $B( \init \ot \Khat)$ instead;
the two are, of course, equivalent.
\end{rems}

Henceforth we focus on the repeated-interactions model of~\cite{AtP}.
Recall that in this model one has a family of discrete-time evolutions
of an open quantum system consisting of a system $\mathsf{S}$ coupled
to a heat reservoir modeled by an infinite chain of identical particles
in some (thermal) state $\rho$, repeatedly interacting with the system
over a short time period of length $\tau$. The corresponding discrete-time
evolution has unitary generator
\[
\exp{\ri \tau \Htot( \tau)},
\]
where the \emph{total Hamiltonian}
decomposes as
\[
\Htot( \tau) =
I_{\wh{\Kil}} \otimes \Hsys + %
\Hpar \otimes I_\init + \tau^{-1 / 2} \Hint
\in B(\Khat \ot \init)
\]
for a \emph{system Hamiltonian}
$\Hsys \in B(\init)_{\sa}$,
a \emph{particle Hamiltonian}
$\Hpar \in B(\wh{\Kil})_{\sa}$ and
an \emph{interaction Hamiltonian} $\Hint \in B(\Khat \ot \init)_{\sa}$.
The continuous limit of this model (embedded into Boson Fock space)
at zero temperature is captured by the following theorem
in which, for each $\tau > 0$,
$( U(\tau)_n )_{n \ges 0}$ denotes the quantum random walk with unitary generator
$\exp{\ri \tau \Htot( \tau)}$.

\begin{thm}\label{thm: 5.5}
Suppose that
$( \bra{ \vac } \ot I_\init ) \Hint ( \ket{ \vac } \ot I_\init ) = 0$,
so that
$\ri \Hint ( \ket{ \vac } \ot I_\init ) \in B( \init; \Khat \ot \init )$
has block-matrix form $\smallcol{0}{L}$ for some
$L \in B( \init; \Kil \ot \init )$. For all $\tau > 0$, set
\[
U^\tau := %
\bigl( ( I_\init \ot D_\tau ) U(\tau)_{\lfloor t/\tau \rfloor } %
( I_\init \ot D_\tau )^* \bigr)_{t\ges 0}
\]
and
\[
F := %
\gaussintegrand{\ri \Hsys + \ri %
\ip{ \vac }{ \Hpar \, \vac } I_\init - \half L^* L}%
{ L }{ - L^* } \in  B( \Khat \otimes \init)_0.
\]
Then $F^* + F + F^* \Delta F = 0$ and, as $\tau \to 0+$,
\[
\sup_{t \in [0,T]} \norm{ ( U^\tau_t - Y^F_t ) \xi } \to 0
\qquad \text{for all } \xi \in \init \ot \Fock %
\text{ and } T \in \Rplus,
\]
where $Y^F$ is a Gaussian HP~cocycle with stochastic generator $F$, as
in Definition~\ref{defn: 2.10}.
\end{thm}
\begin{proof}
That $F$ satisfies the structure relation is readily verified.
The final claim holds
by \cite[Theorem 7.6 and Remarks 4.8 and~5.10]{B1}
(see also \cite[Theorem 4.3]{BGL}), since
\[
\lim_{\tau \to 0+~} %
\tauscale\bigl( \exp{\ri \tau \Htot( \tau)} - I_{\init\ot\Fock} \bigr) = F.
\qedhere
\]
\end{proof}

\subsection*{Particles in a faithful state}

We now fix a non-zero Hilbert space $\particle$, referred to as the
\emph{particle space}, and a faithful normal state $\rho$ on
$B(\particle)$. Let $(\gamma_\alpha)_{\alpha \ges 0}$ be the
eigenvalues of its density matrix~$\vrho$,
ordered to be strictly decreasing,
 the index set being either~$\Zplus$
or~$\{ 0, 1, \cdots , N \}$ for some non-negative integer $N$.
For any index $\alpha$,  let $P_\alpha \in B(\particle)$
be the orthogonal projection with range $\ka$, the eigenspace of
$\vrho$ corresponding to the eigenvalue $\gamma_\alpha$. Thus
$\vrho = \sum_{\alpha \ges 0} \gamma_\alpha P_\alpha$
and
$\sum_{\alpha \ges 0} \gamma_\alpha d_\alpha  = 1$,
where
$d_\alpha := \dim \kil_\alpha = \tr( P_\alpha )$.

Let $( \Khat, \pi, \eta)$ denote the GNS representation of $\rho$.
Thus $( \pi, \Khat)$ is a normal unital $*$-representation of
$B(\particle)$, $\eta$ is an operator from $B(\particle)$ to $\Khat$
with dense range,
\[
\pi(X) \eta(Y) = \eta(XY)
\quad \text{and} \quad
\ip{ \eta(Z) }{ \pi(X) \eta(Y) } = \rho (Z^* X Y )
\quad \text{for all } X,Y,Z \in B(\particle).
\]
In particular, $\rho(X) = \ip{ \omega }{ \pi(X) \omega }$ and
$\eta(X) = \pi(X) \vac$ for all $X \in B(\particle)$, where
$\omega := \eta(I_\particle)$. As is well known, the GNS
representation is unique up to isomorphism; here, we take the triple
defined as follows:
\[
\Khat := HS( \particle ), \qquad
\pi(X) := L_X \quad \text{and} \quad %
\eta(X) := X \vrho^{1/2} =
\sum_{\alpha \ges 0} \sqrt{ \gamma_\alpha } \, X P_\alpha
\qquad \text{for all } X \in B(\particle),
\]
where $HS( \particle )$ denotes the Hilbert--Schmidt class of
operators on~$\particle$ and $L_X$ denotes the operator of left
multiplication by~$X$. In particular, $\omega = \vrho^{1/2}$. Now let
\[
\Kil := \Khat \ominus \Comp \omega, \qquad %
\pitilde := \pi \uwot \id_{B( \init )} \quad \text{and} \quad %
\rhotilde := \rho \uwot \id_{B( \init )},
\]
so that
$( \pitilde, \Khat \ot \init )$
is a normal unital $*$-representation of $B( \particle \ot \init )$
and
$\rhotilde$ is a normal unital completely positive map from
$B( \particle \ot \init )$ to $B( \init )$. For all
$\alpha, \beta \ges 0$, let
\[
\kab := \Lin\bigl\{ \dyad{ x }{ y }: \, x \in \ka, y \in \kb \bigr\},
\]
and let
\[
\noise := \bigoplus_{\alpha > \beta \ges 0} \kab, \qquad
\kbar := \bigoplus_{0 \les \alpha < \beta} \kab, \qquad
\Kil_1 := \noise \op \kbar \quad \text{and} \quad %
\Kil_0 := \Khat \ominus ( \Comp \vac \op \Kil_1 ).
\end{equation*}
Let $k$ be the anti-unitary operator from $\noise$ to $\kbar$  obtained
by restricting the adjoint operation on $\Khat = HS( \particle )$.
Then
\[
\Khat = \Comp \omega \op \Kil_1 \op \Kil_0,
\]
and $( \kbar, k )$ is a realisation of the conjugate Hilbert space
of~ $\noise$. Note also that
\begin{equation}\label{eqn: C vac zero}
\Comp \omega  \op \Kil_0 = \bigoplus_{\alpha \ges 0} \kaa.
\end{equation}
We now identify the one-dimensional subspace $\Comp \vac$ of $\Khat$
with $\Comp$, so that
\[
\Khat = \wh{ \Kil_1 } \op \Kil_0, \quad \text{where } %
\wh{ \Kil_1 } = \Comp \op \Kil_1, \qquad \text{and} \qquad
\vac = \begin{pmatrix} 1 \\ 0 \\ 0 \end{pmatrix}.
\]

\begin{thm}\label{thm: 5.11}
Let the operators $\Hsys \in B(\init)$, $\Hpar \in B( \particle )$
and $\Hint \in B( \particle \ot \init )$ be self-adjoint, and assume
that $( \Pa \ot I_\init ) \Hint  ( \Pa \ot I_\init ) = 0$ for all
$\alpha \ges 0$. Then we have the following.
\begin{alist}
\item The operator
$\pitilde( \ri \Hint ) ( \ket{ \vac } \ot I_\init ) \in %
B( \init; ( \wh{ \Kil_1 } \op \Kil_0 ) \ot \init )$
has the block-matrix form
$\left[\begin{smallmatrix}0 \\ L \\ 0 \end{smallmatrix}\right]$
for some $L \in B( \init;  \Kil_1  \ot \init )$.
\item  For all $\tau > 0$, let
$\Utilde^\tau := %
\bigl( ( I_\init \ot D_\tau ) %
\Utilde(\tau)_{\lfloor t/\tau \rfloor } ( I_\init \ot D_\tau )^* %
\bigr)_{t \ges 0}$, where $( \Utilde( \tau )_n )_{n \ges 0}$
is the quantum random walk generated by
$\pitilde\bigl( \exp \ri \tau \Htot ( \tau ) \bigr)$ and
\begin{equation*}
\Htot ( \tau )  :=  %
I_{\particle} \otimes \Hsys + \Hpar \otimes I_\init + %
\tau^{-1 / 2} \Hint \in  B( \particle \otimes \init),
\end{equation*}
and let $\Ftilde := F \op 0_{\Kil_0 \ot \init }$,
where
\[
F := \gaussintegrand{ K }{ L }{ - L^* } \in  %
B( \wh{ \Kil_1 } \otimes \init ) \qquad \text{with } %
K := \ri \Hsys + \ri \rho( \Hpar ) I_\init - %
\half \rhotilde ( \Hint^2 ).
\]
Then
$\Ftilde^* + \Ftilde + \Ftilde^* \Delta \Ftilde = 0$
and, as $\tau \to 0+$,
\[
\sup_{t \in [0,T]} %
\bigl\| ( \Utilde^\tau_t - Y^{\Ftilde}_t ) \xi \bigr\| \to 0 %
\qquad \text{for all } \xi \in \init \ot \Fock \text{ and } T \in
\Rplus.
\]
\end{alist}
\end{thm}

\begin{proof}
(a) It must be shown that
\begin{equation}\label{eqn: RanTilde} 
\Ran \pitilde ( \Hint ) ( \ket{ \vac } \ot I_\init )
\perp ( \Comp \vac \op \Kil_0 ) \ot \init.
\end{equation}
If $u$, $v \in \init$ and $T \in \kaa$ for some index $\alpha$, and
$\Huv := ( I_\particle \ot \bra{ u } ) \Hint %
( I_\particle \ot \ket{ v } )$,
then
\begin{align*}
\Bigl\langle T \ot u, %
\pitilde( \Hint ) ( \vac \ot v ) \Bigr\rangle = %
\bigl\langle T, \Huv \vac \bigr\rangle = %
\sqrt{ \gamma_\alpha } \, %
\bigl\langle T, \Huv P_\alpha \bigr\rangle = 0,
\end{align*}
and so~(\ref{eqn: RanTilde}) follows from~(\ref{eqn: C vac zero}).

(b) Note that
$\pitilde\bigl( \exp \ri \tau \Htot ( \tau ) \bigr) =  %
\exp \ri \tau  \wt{\Htot}( \tau ) \in %
B( \Khat \otimes \init)$
for all $\tau > 0$, where
\[
\wt{\Htot}( \tau ) := I_{\Khat} \otimes \Hsys + %
\pi( \Hpar ) \otimes I_\init + \tau^{-1 / 2} \pitilde( \Hint ).
\]
Furthermore
$\ip{ \vac }{ \pi ( \Hpar ) \vac } I_\init = \rho( \Hpar ) I_\init$
and it is straightforward to verify that
\[
L^*L = \begin{bmatrix} 0 \\ L \\ 0 \end{bmatrix}^* %
\begin{bmatrix} 0 \\ L \\ 0 \end{bmatrix} = %
( \bra{ \vac } \ot I_\init ) \pitilde( \Hint^2 ) %
( \ket{ \vac } \ot I_\init ) = \rhotilde ( \Hint^2 ),
\]
which implies that $F$ is as claimed. The conclusion
now follows from Theorem~\ref{thm: 5.5}, since $\vac$ is identified
with
$\left( \begin{smallmatrix} 1 \\ 0 \\ 0 \end{smallmatrix}\right) %
\in \wh{ \Kil_1 } \ot \Kil_0$
and $\Comp \vac$ with $\Comp$.
\end{proof}

\begin{rems}
Under the identification
$\init \ot \Fock^{\Kil} = %
\init \ot \Fock^{\Kil_1} \ot \Fock^{\Kil_0}$,
where $\Fock^\Hil := \Gamma\bigl( L^2( \Rplus; \Hil ) \bigr)$,
the limit process decomposes as
\[
Y^{\Ftilde}_t = \YF_t \ot I_{\Fock^{\Kil_0}}
\qquad \text{for all } t \in \Rplus.
\]

The condition on $\Hint$ has the following physical interpretation:
there is no contribution from the interaction Hamiltonian unless the
particle undergoes a transition.
\end{rems}

\begin{ass}\label{ass: 5.13}
We impose an exponential-decay condition on the
eigenvalues of the density matrix $\vrho$, by insisting that
\[
m_\rho := %
\inf \{ \gamma_\alpha / \gamma_{\alpha + 1}: \alpha \ges 0 \} > 1.
\]
\end{ass}

This ensures that the following lemma yields an AW amplitude for
$\noise$. To avoid it would require more of the general theory
developed in [$\text{LM}_{1,2}$].

For all indices $\alpha$ and $\beta$, let $\Pab$ denote the orthogonal
projection with range $\kab$.

\begin{lemma}\label{lemma: 5.14}
Suppose the state $\rho$ satisfies Assumption~\ref{ass: 5.13}.
Define an operator
\begin{equation}\label{eqn: PreT} 
S( \rho ) := \st.\sum_{\alpha > \beta \ges 0} %
\sqrt{\tfrac{ \gamma_\alpha }{ \gamma_\beta - \gamma_\alpha }} \, %
\Pab \in B( \noise )_+,
\end{equation}
where the series converges in the strong sense, and, let
$C(\rho) := \sqrt{ I_\noise + S(\rho)^2 }$, then
\begin{equation}\label{eqn: RootTbar}
C(\rho) = \st.\sum_{\alpha > \beta \ges 0} %
\sqrt{\tfrac{ \gamma_\beta }{\gamma_\beta - \gamma_\alpha}} \, %
\Pab \quad \text{and} \quad  \ol{S(\rho)} = %
\st.\sum_{\alpha > \beta \ges 0} %
\sqrt{\tfrac{ \gamma_\alpha }{ \gamma_\beta - \gamma_\alpha }} \, %
\Pba.
\end{equation}
\end{lemma}
\begin{proof}
If $\alpha, \beta \ges 0$ with $\alpha > \beta$, and
$\zeta\in \kil_{\alpha'}$ and $\eta \in \kil_{\beta'}$ with
$\alpha'$, $\beta' \in I$, then
\begin{subequations}
\label{eqn: 5.6}
\begin{align}
& \label{eqn: 5.6a}
0 \les \tfrac{ \gamma_\alpha }{  \gamma_\beta - \gamma_\alpha } = %
\big( \tfrac{ \gamma_\beta }{ \gamma_\alpha } - 1 \big)^{-1} \les %
\big( \tfrac{ \gamma_\beta }{ \gamma_{\beta + 1} } - 1 \big)^{-1} %
\les( m_\rho - 1 )^{-1}, \\
&\label{eqn: 5.6b}
1 + \tfrac{ \gamma_\alpha }{ \gamma_\beta - \gamma_\alpha } = %
\tfrac{ \gamma_\beta }{ \gamma_\beta - \gamma_\alpha } \\
\text{and} \quad & \label{eqn: 5.6c}
\ol{ \Pab } \big( \dyad{ \eta }{ \zeta } \big) = %
\big( \Pab  \big( \dyad{ \zeta }{ \eta } \big) \big)^* = %
\delta_{\alpha \alpha'} \delta_{\beta \beta'} %
\dyad{ \eta }{ \zeta } = %
\Pba  \big( \dyad{ \eta }{ \zeta } \big),
\end{align}
\end{subequations}
where $\delta$ is the Kronecker delta. From~(\ref{eqn: 5.6a}) it
follows that~(\ref{eqn: PreT}) defines a non-negative bounded operator
$S(\rho)$ on $\noise$, and from~(\ref{eqn: 5.6c}) it follows that
$\ol{ \Pab } = \Pba$ for all $\alpha > \beta \ges 0$, so the
identities~(\ref{eqn: RootTbar}) follow from~(\ref{eqn: 5.6b})
and~(\ref{eqn: 5.6c}).
\end{proof}

Thus, under Assumption~\ref{ass: 5.13}, with $S(\rho)$ and $C(\rho)$
as in the preceding lemma,
\begin{equation}
\label{eqn: Srho}
\Sigma(\rho) := \mat{ C(\rho) }{ 0 }{ 0 }{ \ol{ S(\rho) } }
\in
B( \noise \op \kbar )
\end{equation}
defines a gauge-invariant AW~amplitude for $\noise$.

Our goal now is to prove that the HP~cocycle generated by $F$ in
Theorem~\ref{thm: 5.11} is $\Sigma(\rho)$-quasifree, provided that the
interaction Hamiltonian $\Hint$ is $\particle$-conjugatable.  To this
end, note first that, for all $T \in B( \particle )$ and all indices
$\alpha, \beta, \alpha'$ and $\beta'$,
$P_\alpha T P_\beta$ and $P_{\alpha'} T P_{\beta'}$ are orthogonal vectors in
$HS( \particle )$ unless $\alpha' = \alpha$ and $\beta' = \beta$, and
therefore
\[
\sum_{\alpha > \beta \ges 0} ( \gamma_\beta - \gamma_\alpha ) %
\norm{ P_\alpha T P_\beta }_2^2 \les %
\sum_{\beta \ges 0} \gamma_\beta \norm{ T P_\beta }_2^2 \les %
\norm{ T }^2 \sum_{\beta \ges 0} \gamma_\beta d_\beta = \norm{ T }^2
\]
and
\[
\sum_{\alpha > \beta \ges 0} ( \gamma_\beta - \gamma_\alpha ) %
\norm{ P_\beta T P_\alpha }_2^2 = \sum_{\alpha > \beta \ges 0} %
( \gamma_\beta - \gamma_\alpha ) \norm{ P_\alpha T^* P_\beta }_2^2 %
\les \norm{ T^* }^2 = \norm{ T }^2,
\]
so the following prescriptions define bounded operators:
\begin{align*}
\phi_\rho & : B( \particle ) \to \ket{ \noise }; \ %
T \mapsto \sum_{\alpha > \beta \ges 0} %
\sqrt{ \gamma_\beta - \gamma_\alpha } \ket{ P_\alpha T P_\beta } \\
\text{and} \qquad
\ol{\phi}_\rho & : B( \particle ) \to \ket{ \kbar }; \ %
T \mapsto \sum_{\alpha > \beta \ges 0} %
\sqrt{ \gamma_\beta - \gamma_\alpha } \ket{ P_\beta T P_\alpha }.
\end{align*}

For the next proposition we adopt the notation
\begin{equation}\label{eqn: adj}
\Bcparticle ( \particle \ot \init )^* := %
\bigl\{ A^*: A \in \Bcparticle ( \particle \ot \init ) \bigr\}.
\end{equation}
Recall that Theorem~\ref{thm: 3.28} gives the inclusion
$HS( \particle ) \otalg B( \init ) \subseteq %
\Bcparticle ( \particle \ot \init )$.
We will show that the maps
$\phi_\rho |_{HS( \particle )} \otalg \id_{ B( \init ) }$
and
$\ol{\phi}_\rho |_{HS( \particle )} \otalg \id_{ B( \init ) }$
extend to operators from
$\Bcparticle ( \particle \ot \init )^*$
to~$B(\init; \noise \ot \init)$ and from
$\Bcparticle ( \particle \ot \init )^*$ to
$B(\init; \kbar \ot \init)$, respectively, and that the resulting maps
are related via partial conjugation.

\begin{propn}\label{propn: phirho}
There are unique operators
\[
\phiinitrho: \Bcparticle( \particle \ot \init )^* \to %
B( \init; \noise \ot \init ) \qquad \text{and} \qquad %
\phibarinitrho: \Bcparticle( \particle \ot \init )^* \to %
B( \init; \kbar \ot \init )
\]
such that
\begin{subequations}
\begin{align}
&
\label{eqn: adj a}
\bigl\langle \dyad{ \zeta }{ \eta } \ot u, %
\phiinitrho ( A ) v \bigr\rangle = %
\sqrt{ \gamma_\beta - \gamma_\alpha } \,
\bigl\langle \zeta \ot u , A ( \eta \ot v ) \bigr\rangle \\
\text{and} \qquad &
\label{eqn: adj b}
\bigl\langle \dyad{ \eta }{ \zeta } \ot u,  %
\phibarinitrho ( A ) v \bigr\rangle = %
\sqrt{ \gamma_\beta - \gamma_\alpha } \, %
\bigl\langle \eta \ot u , A ( \zeta \ot v ) \bigr\rangle
\end{align}
\end{subequations}
for all
$A \in \Bcparticle( \particle \ot \init )^*$, $u$, $v \in \init$
and $\zeta \in \noise_\alpha$, $\eta \in \noise_\beta$ with
$\alpha > \beta$. Furthermore, we have that
\[
\norm{ \phiinitrho ( A ) } \les c( A^* )
\quad \text{ and } \quad %
\norm{ \phibarinitrho ( A ) } \les c( A^* ) %
\qquad \text{for all } A \in \Bcparticle ( \particle \ot \init )^*,
\]
and the following properties hold.
\begin{alist}
\item
If $A \in \Bcparticle ( \particle \ot \init )^* \cap %
\Bcparticle ( \particle \ot \init )$ then
$\phiinitrho ( A )$ is $\noise$-conjugatable,
$\phibarinitrho ( A^* )$ is $\ol{\noise}$-conjugatable
and~$\phiinitrho ( A )^{\pc} = \phibarinitrho ( A^* )$.
Thus
\[
c\bigl( \phiinitrho( A ) \bigr)  = %
\norm{ \phibarinitrho ( A^* ) } \les c( A ) %
\qquad \text{ and } \qquad
c\bigl( \phibarinitrho ( A^* ) \bigr) = %
\norm{ \phiinitrho ( A ) } \les c( A^* ).
\]
\item The maps $\phiinitrho$ and $\phibarinitrho$ are extensions
of $\phi_\rho |_{HS( \particle )} \otalg \id_{ B( \init ) }$ and
$\ol{\phi}_\rho |_{HS( \particle )} \otalg \id_{ B( \init ) }$,
respectively.
\end{alist}
\end{propn}

\begin{proof}
Uniqueness is clear. For existence, let
$A \in \Bcparticle ( \particle \ot \init )^*$.
For each $\alpha \ges 0$, choose an orthonormal basis
$\big( e^i_\alpha \big)_{i=1}^{d_\alpha}$ for $\noise_\alpha$, and
note that if $u \in \init$ then
\[
\sum_{\alpha > \beta \ges 0} ( \gamma_\beta - \gamma_\alpha ) \, %
\sum_{i=1}^{d_\alpha} \sum_{j=1}^{d_\beta} \bigl\| %
( \bra{ \eia } \ot I_\init ) A ( \ejb \ot u ) \bigr\|^2 \les %
\sum_{ \beta \ges 0 } \gamma_\beta \, \sum_{j=1}^{d_\beta}
\big\|  A ( \ejb \ot u ) \big\|^2 \les %
c( A^* ) \norm{ u }^2.
\]
Hence
\[
u \mapsto %
\sum_{\alpha > \beta \ges 0} \sqrt{ \gamma_\beta - \gamma_\alpha } \, %
\sum_{i=1}^{d_\alpha} \sum_{j=1}^{d_\beta} \ket{ \eia } \bra{ \ejb } %
\ot \bigl( \bra{ \eia } \ot I_\init \bigr) A ( \ejb \ot u )
\]
defines an operator  $\phiinitrho (A)$ from $\init$ to
$\noise \ot \init$ such that
$\norm{ \phiinitrho ( A ) } \les c( A^* ) $; it also
satisfies~(\ref{eqn: adj a}) since, for all $u$, $v \in \init$,
$\zeta \in \noise_\alpha$ and $\eta \in \noise_\beta$, where
$\alpha > \beta$,
\begin{align*}
\bigl\langle \dyad{ \zeta }{ \eta } \ot u, %
\phiinitrho ( A ) v \big\rangle & = %
\sqrt{ \gamma_\beta - \gamma_\alpha } \, %
\bigl\langle u , \bigl( \bra{ \zeta } \ot I_\init \bigr) %
A \bigl( \ket{ \eta } \ot I_\init \bigr) v  \bigr\rangle \\
 & = \sqrt{ \gamma_\beta - \gamma_\alpha } \, %
\bigl\langle \zeta \ot u , A ( \eta \ot v ) \bigr\rangle.
\end{align*}
In particular, the operator $\phiinitrho ( A )$ does not depend on the
choice of orthonormal bases made above.  Similarly, there is an
operator $\phibarinitrho (A)$ from $\init$ to $\kbar \ot \init$
such that $\norm{ \phibarinitrho ( A ) } \les c( A^* )$, the
identity~(\ref{eqn: adj b}) holds and, for any choice of orthonormal
bases $\bigl( e^i_\alpha \bigr)_{i=1}^{d_\alpha}$ for $\noise_\alpha$,
\[
\phibarinitrho ( A ) u = \sum_{\alpha > \beta} %
\sqrt{ \gamma_\beta - \gamma_\alpha } \, %
\sum_{i=1}^{d_\alpha} \sum_{j=1}^{d_\beta} \dyad{ \ejb }{ \eia } \ot %
\bigl( \bra{ \ejb } \ot I_\init \bigr) A ( \eia \ot u ).
\]
(a) If
$A \in \Bcparticle ( \particle \ot \init )^* \cap %
\Bcparticle ( \particle \ot \init )$, $u$, $v \in \init$ and
$\zeta \in \noise_\alpha$, $\eta \in \noise_\beta$ with
$\alpha > \beta$, then
\begin{align*}
\bigl\langle \phiinitrho ( A ) u, %
\dyad{ \zeta }{ \eta } \ot v \bigr\rangle & = %
\sqrt{ \gamma_\beta - \gamma_\alpha } \, %
\bigl\langle A ( \eta \ot u ), \zeta \ot v \bigr\rangle \\
 & = \sqrt{ \gamma_\beta - \gamma_\alpha } \, %
\bigl\langle \eta \ot u , A^* ( \zeta \ot v ) \bigr\rangle = %
\bigl\langle \dyad{ \eta }{ \zeta } \ot u, %
\phibarinitrho ( A^* ) v \bigr\rangle.
\end{align*}
Therefore, by linearity,
\[
\bigl\langle \phiinitrho ( A ) u, T \ot v \bigr\rangle = %
\bigl\langle \ol{ T } \ot u, \phibarinitrho ( A^* ) v \bigr\rangle %
\qquad \text{for all } u, v \in \init \text{ and } T \in \noise,
\]
so $\phiinitrho ( A )$ is $\noise$-conjugatable and
$\phiinitrho ( A )^\pc = \phibarinitrho ( A^* )$.

(b) Let $T \in HS( \particle )$ and $X \in B( \init )$. Then
$T \ot X \in \Bcparticle ( \particle \ot \init )^* \cap %
\Bcparticle ( \particle \ot \init )$,
by Theorem~\ref{thm: 3.28}. Comparing matrix elements, the identities
\[
\phi_\rho ( T ) \ot X = \phiinitrho ( T \ot X )  \quad \text{ and } \quad %
\ol{\phi}_\rho ( T ) \ot X = \phibarinitrho ( T \ot X )
\]
are readily verified, so (b) follows.
\end{proof}

Recall that a countable family of bounded operators $\mathcal{C}$ is
said to be \emph{strongly linearly independent} if there is no
non-zero function $\alpha: \mathcal{C} \to \Comp$ such that
$\sum_{T \in \mathcal{C} } \alpha ( T ) T$ converges to zero in the
strong sense.

\begin{thm}
\label{thm: 5.15}
Let $\Hsys \in B(\init)_{\sa}$, $\Hpar \in B( \particle )_{\sa}$ and
$\Hint \in B( \particle \ot \init )_{\sa}$, where
$( \Pa \ot I_\init ) \Hint  ( \Pa \ot I_\init ) = 0$ for all
indices $\alpha$ and, as in Theorem~\ref{thm: 5.11}, set
$F = \left[\begin{smallmatrix} K & - L^* \\
L  & 0 \end{smallmatrix}\right] \in
B\big( \wh{ \Kil_1  } \ot \init \big)_0$ where
\[
\Kil_1 = \noise \op \kbar, \quad %
K = \ri \Hsys + \ri \rho( \Hpar ) I_\init - %
\half \rhotilde ( \Hint^2 ) \quad \text{and} \quad %
L = \ri J^* \wt{ \pi }( \Hint ) %
\bigl( \ket{ \omega } \ot I_\init \bigr);
\]
$J$ being the natural isometry from $\noise \op \kbar$ to
$\Comp \op ( \noise \op \kbar ) \op \Kil_0$.

Suppose that the state $\rho$ satisfies Assumption~\ref{ass: 5.13},
and the operator $\Hint \in B( \particle \ot \init )$ is
$\particle$-conjugatable. Then the  HP~cocycle $U := \YF$ is
$\Sigma(\rho)$-quasifree with  $\Sigma(\rho)$-generator
$\smallgaussintegrand{K}{ Q }{ -Q^* }$, where
$Q = \ri \phiinitrho ( \Hint )$.

Suppose further that, with respect to some orthonormal bases
$\big( e^i_\alpha \big)_{i=1}^{d_\alpha}$ for each $\noise_\alpha$,
the family
$\bigl\{ ( \bra{ \eia } \ot I_\init ) \Hint %
( \ket{ \ejb } \ot I_\init ) : \alpha > \beta \ges 0 \ %
i = 1, \ldots , d_\alpha, \  j = 1, \cdots , d_\beta \bigl\}$
is strongly linearly independent. Then $\Sigma(\rho)$ is the unique
gauge-invariant AW~amplitude with respect to which the HP~cocycle $U$ is quasifree.
\end{thm}

\begin{proof}
Since $U$ is a Gaussian HP~cocycle with stochastic generator
$\smallgaussintegrand{K}{ L}{ - L^* }$,
Corollary~\ref{cor: 3.11} implies that, for the first part, it
suffices to verify the identity
\[
L = ( \Sigma ( \rho ) \ot I_\init ) %
\begin{bmatrix} Q \\ - Q^\pc \end{bmatrix}.
\]
Since $\Hint$ is a self-adjoint $\particle$-conjugatable operator, by
assumption, Proposition~\ref{propn: phirho} ensures that the
operators~$Q$ and~$R := \ri \phibarinitrho ( \Hint )$ are well defined
and conjugatable, with
$Q^\pc = -\ri \phiinitrho ( \Hint )^\pc = -R$. Thus,
for all $\zeta \in \noise_\alpha$ and
$\eta \in \noise_\beta$ with $\alpha > \beta$, and
all $u$, $v \in \init$,
\[
\bigl\langle \dyad{ \zeta }{ \eta } \ot u, Q v \bigr\rangle = %
\ri \sqrt{ \gamma_\beta - \gamma_\alpha } \, %
\bigl\langle \zeta , \Huv \eta  \bigr\rangle %
\quad \text{and} \quad %
\bigl\langle \dyad{ \eta }{ \zeta } \ot u, Q^\pc v \bigr\rangle = %
-\ri \sqrt{ \gamma_\beta - \gamma_\alpha } \, %
\bigl\langle \eta , \Huv \zeta  \bigr\rangle,
\]
where
$\Huv := %
( I_\particle \ot \bra{ u } ) \Hint ( I_\particle \ot \ket{ v } )$.
Hence,
by Lemma~\ref{lemma: 5.14},
\[
\bigl\langle \dyad{ \zeta }{ \eta } \ot u, %
C( \rho ) Q v \bigr\rangle = \ri \sqrt{ \gamma_\beta } \, %
\bigl\langle \zeta , \Huv \eta \bigr\rangle %
\quad \text{and} \quad %
\bigl\langle \dyad{ \eta }{ \zeta } \ot u, %
\ol{ S(\rho) } Q^\pc v \bigr\rangle = %
-\ri \sqrt{ \gamma_\alpha } \, %
\bigl\langle \eta , \Huv \zeta \bigr\rangle.
\]
On the other hand, by definition, the operator $L$ is such that
\[
\ip{ \chi \ot u }{ L v } = \ri \ip{ \chi }{ \Huv \omega } \qquad %
\text{for all } u, v \in \init \text{ and } %
\chi \in \noise \op \kbar.
\]
Thus, in terms of the block-matrix decomposition
$L = \smallcol{ L_1 }{ L_2 } \in %
B( \init; ( \noise \op \kbar ) \ot \init )$,
\[
L = %
\col{ ( C( \rho ) \ot I_\init ) Q }%
{ - ( \ol{ S( \rho ) } \ot I_\init ) Q^\pc } = %
( \Sigma ( \rho ) \ot I_\init ) \col{Q}{ - Q^\pc },
\]
as required.

Finally, for each index $\alpha$,
let $( \eia )_{i=1}^{d_\alpha}$ be an orthonormal
basis for $\ka$
and,
for indices $\alpha$ and $\beta$, set
$\eijab := \dyad{ \eia }{ \ejb }$ for  all
$i = 1$, \ldots, $d_\alpha$ and $j = 1$, \ldots, $d_\beta$. Then,
for all~$x \in \noise \setminus \{ 0 \}$, the family
$\bigl\{ \ip{ x }{ \eijab } : %
\alpha > \beta \ges 0, \ i = 1, \ldots, d_\alpha, \ %
j = 1, \ldots , d_\beta \bigr\}$ is not identically zero
and so, under the strong linear independence assumption,
\[
( \bra{ x } \ot I_\init ) L_1 = \st.\sum_{\alpha > \beta \ges 0} %
\sum_{i=1}^{d_\alpha} \sum_{j=1}^{d_\beta} \sqrt{ \gamma_\beta } \, %
\ip{ x }{ \eijab } \bigl( \bra{ \eia } \ot I_\init \bigr) %
\ri H_I \bigl( \ket{ \ejb } \ot I_\init \bigr)  \neq 0.
\]
In other words $\noise^{L_1} = \{ 0 \}$ and therefore,
by Corollary~\ref{cor: XiAQ}, there is no other gauge-invariant
AW~amplitude $\Sigma$ with respect to which the HP~cocycle $U$ is
$\Sigma$-quasifree.
\end{proof}

\begin{rem}
Theorems~\ref{thm: 5.11} and~\ref{thm: 5.15} comprise a significant
generalisation of the main result of~\cite[Theorem 7]{AtJ}. The
restriction to finite-dimensional noise or particle space, is removed,
and the interaction Hamiltonian is of a more general form. In
\cite{AtJ}, the operator~$\Hint$ is taken to have the form
$\smallgaussintegrand{0}{V}{V^*}$ so that $\ri \Hint$ is of the above form
with $L = \ri V$. This assumption corresponds to
the $\Sigma(\rho)$-quasifree generator $\smallQFintegrand{K}{Q}{-Q^*}$
satisfying
\[
\bigl( \bra{ \ejk } \ot I_\init \bigr) Q = 0 \qquad
\text{for all } j > k > 0.
\]
In conclusion, a large class of unitary quantum random walks, with
particles in a faithful normal state, converge to HP~cocycles governed
by a quasifree quantum Langevin equation.

The results in this section could be applied to bipartite systems,
as studied in \cite{ADP}, in non-zero temperature.
In this model, two non-interacting quantum systems
are both coupled to an environment comprising an infinite chain of
identical and independent particles, with each particle now in the same
non-zero temperature state.
For the zero temperature case see
\cite[Theorem 3.1]{ADP} and \cite[Theorem 8.2]{BGL};
the methods developed in \cite{BGL} adapt nicely to the quasifree context.
\end{rem}

\setcounter{section}{0}
\renewcommand{\theequation}{\Alph{section}.\arabic{equation}}
\renewcommand{\thesection}{\Alph{section}}
\renewcommand{\thepropn}{\Alph{section}.\arabic{propn}}
\section*{Appendix}
\setcounter{section}{1} \setcounter{propn}{0} \setcounter{equation}{0}

In this appendix, we prove that symplectic automorphisms of a Hilbert
space~$\hil$ are necessarily bounded, and give a parameterisation for
the elements of the group $S(\hil)^\times$.  For the convenience
of the reader, this is a streamlined version of the proof given
in~\cite{HoR}, which also covers the case of unbounded symplectic
automorphisms of separable pre-Hilbert spaces.

\begin{propn}\label{propn: B bdd}
Let $B \in S(\hil)^\times$. Then $B$ is bounded.
\end{propn}

\begin{proof}
Let $L$ and $A$ be the linear and conjugate-linear parts of $B$, as in
(\ref{eqn: LA}). For all $z$, $x \in \hil$,
\begin{align*}
2 \ip{ Lz }{ x } & = %
\bigl\langle B z - \ri B( \ri z), x \bigr\rangle \\
 & = \re \ip{ Bz }{ x } + \ri \im \ip{ Bz }{ x } + %
\ri \re \ip{ B( \ri z ) }{ x } - \im \ip{ B(\ri z ) }{ x } \\
 & = \im \ip{ Bz }{ \ri x } + \ri \im \ip{ z }{ B^{-1} x } + %
\ri \im \ip{ B( \ri z ) }{ \ri x } - \im \ip{ \ri z }{ B^{-1} x } \\
 & = \im \ip{ z }{ B^{-1}(\ri x) } + \ri \im \ip{ z }{ B^{-1} x } + %
\ri \im \ip{ \ri z }{ B^{-1}( \ri x ) } + \re \ip{ z }{ B^{-1} x } \\
 & = \ip{ z }{ B^{-1} x } - \ri \bigl( %
\re \ip{ z }{ B^{-1}( \ri x ) } +  \ri \im \ip{ z }{ B^{-1}(\ri x) } %
\bigr) \\
 & = \ip{ z }{ B^{-1} x } - \ri  \ip{ z }{ B^{-1}( \ri x ) }.
\end{align*}
Thus $L$ has everywhere-defined adjoint
$x \mapsto \half \bigl( B^{-1} x - \ri B^{-1} ( \ri x ) \bigr)$,
and so is closed, and therefore bounded, by the closed graph theorem.
Similarly, the conjugate-linear operator $A$ has everywhere-defined
adjoint
$x \mapsto -\half \bigl( B^{-1} x + \ri B^{-1} ( \ri x ) \bigr)$,
and so is also bounded. Thus $B$ is bounded.
\end{proof}

For a triple $(V,C,P)$ consisting of a unitary operator $V$ on $\hil$,
a bounded non-negative operator~$P$ on $\hil$ and a conjugation
(a self-adjoint anti-unitary operator) $C$ on $\hil$, such that $P$
and $C$ commute, we define the following bounded real-linear operator
on~$\hil$:
\begin{equation}\label{eqn: BUCP}
\BVCP :=  V( \cosh P - C \sinh P  )
\end{equation}

\begin{rem}
Since, with $(V,C,P)$ as above, the map $-C$ is also a conjugation
on~$\hil$ that commutes with $P$, a deliberate choice is being made
here. The reason for this particular choice is that it eliminates
minus signs elsewhere.
\end{rem}

\begin{thm}\label{thm: A1}
\mbox{}\par
\begin{alist}
\item  Let $(V,C,P)$ be a triple as above.
\begin{rlist}
\item The operator $\BVCP$ is a symplectic automorphism,  with bounded
inverse
\[
( \cosh P + C \sinh P ) V^* = \BVCPinverse.
\]
\item   Suppose that $\BVCP = \BVCPprime$ for another such triple
$(V',C',P')$. Then
\[
V' = V, \quad P' = P \quad %
\text{and $C'$ agrees with $C$ on $\Ran P$}.
\]
\end{rlist}
\item Conversely, let $B \in S( \hil )^\times$. Then there is a
triple $(V,C,P)$ as above, such that $B = \BVCP$.
\end{alist}
\end{thm}
\begin{proof}
(a) (i)
This is readily verified.

(ii) Set $B = \BVCP$, and let $L$ and $A$ be its linear and
conjugate-linear parts. Then
\[
V \cosh P = L = V' \cosh P' \quad \text{and } \quad %
-V C \sinh P = A = -V' C' \sinh P'.
\]
Since the bounded operators $\cosh P$ and $\cosh P'$ are non-negative
and invertible, and $V$ and $V'$ are unitary, the uniqueness of polar
decompositions implies that $V' = V$ and $\cosh P' = \cosh P$.
The non-negativity of $P'$ and $P$ therefore implies that $P' = P$, and thus also
$C' \sinh P' = C \sinh P$. It follows that $C' f( P' ) = C f( P )$ for
all continuous functions $f: \Rplus \to \Comp$ satisfying $f(0) = 0$;
in particular $C' P = C P$, so $C'$ and $C$ agree on $\Ran P$.

(b) Let $L$ and $A$ denote the linear and conjugate-linear parts
of~$B$. It follows from the proof of Proposition~\ref{propn: B bdd}
that $L^*$ and $- A^*$ are respectively the linear and
conjugate-linear parts of $B^{-1}$, so
\[
I_\hil = ( L^* - A^* ) ( L + A ) = %
L^*L - A^*A + L^*A - A^*L.
\]
Therefore, taking linear and conjugate-linear parts,
\begin{equation}
\label{eqn: LIAone}
L^* L = A^*A + I_\hil \qquad\text{and} \qquad L^* A = A^*L.
\end{equation}
Applying the first of these identities to the symplectic automorphism
$B^{-1}$, we see that
\begin{equation}
\label{eqn: LIAtwo}
L L^* = A A^* + I_\hil.
\end{equation}
Let $V | L |$ and $W | A |$ be the polar decompositions of $L$
and~$A$, respectively, and set $\Kil := \Ker | A |$
and~$\Kil^* := \Ker | A^* |$. The conjugate-linear partial isometry
$W$ has initial space $\Kil^\perp$ and final space~$\Kil^{* \perp}$,
and the identities~(\ref{eqn: LIAone}) and~(\ref{eqn: LIAtwo}) imply
that $L$ is invertible, so $V$ is unitary, and~$| L | \ges I_\hil$.
Thus there exists a unique non-negative operator
$P \in B(\hil)$ such that~$| L | = \cosh P$
and~$| A | = ( | L |^2 - I_\hil )^{1/2} = \sinh P$. Now
$| L^* | = V | L | V^*$ and $| L | = V^* | L^* | V$ so,
for all $x \in \Kil$ and~$z \in \Kil^*$,
\[
| L^* | V x =  V | L | x = Vx
\qquad \text{and} \qquad | L | V^* z =  V^* | L^* | z = V^* z,
\]
which implies that
$V \Kil \subseteq \Kil^*$ and $V^* \Kil^* \subseteq \Kil$. Hence
$V \Kil = \Kil^*$, and therefore also $V \Kil^\perp = \Kil^{* \perp}$.
It follows that, on $\hil = \Kil \op \Kil^\perp$,
$V^*W$ has the form $\{ 0 \} \op D_1$
for an anti-unitary operator~$D_1$ on~$\Kil^\perp$.
Therefore, setting $D := D_0 \oplus D_1$
for an arbitrary conjugation $D_0$ on $\Kil$,
\[
B = L + A = V | L | + W | A | = V ( \cosh P + D \sinh P ),
\]
$|A| D = |A| V^*W$ and $V^*W |A| = D |A|$.
Thus,
using the
identities~(\ref{eqn: LIAone}) and~(\ref{eqn: LIAtwo}) once more,
\[
| A | V^* W = ( | L |^2 - I )^{1/2} V^* W = %
V^* ( | L^* |^2 - I )^{1/2} W = V^* | A^* | W = V^* W | A |.
\]
Therefore $D$ commutes with $| A | = \sinh P$ and so commutes with all
continuous functions of $\sinh P$ such as $P$ itself and
$|L|$. The second identity in~(\ref{eqn: LIAone}) now
implies that
\begin{align*}
\ip{ |L| x }{ D |A| y } = \ip{ L x }{ A y } = %
\ip{ L y }{ A x } & = \ip{ |L| y }{ D |A| x } \\
 & = \ip{ |A| y }{ D |L| x } = \ip{ |L| x }{ D^* |A| y } \quad %
\text{for all } x, y \in \hil,
\end{align*}
so $D$ and $D^*$ agree on $\Ranbar |A| = \Kil^\perp$, and thus
$D^*_1 = D_1$. But $D^*_0 = D_0$, since $D_0$ is a conjugation on
$\Kil$, therefore $D^* = D$ and so the anti-unitary operator $D$ is a
conjugation on $\hil$. The proof is now completed by letting $C$ be
the conjugation $-D$.
\end{proof}

{\bf Acknowledgement.} This work was supported by the Leverhulme Trust
Research Project Grant RPG-2014-196
\emph{Quantum random walks and quasi-free quantum stochastic calculus}.


\begin{thebibliography}{GLW}

\bibitem[App]{App}
D.~Applebaum,
Quasi-free stochastic evolutions,
\emph{in}
``Quantum Probability and Applications~II'',
(\emph{eds}.\ L.~Accardi \& W.~von Waldenfels),
\emph{Lecture Notes in Math.}~\textbf{1136},
Springer, Berlin, 1985, pp.~46--56.

\bibitem[ArW]{ArW}
H.~Araki and E.J.~Woods,
Representations of the canonical commutation relations
describing a non\-relativistic infinite free Bose gas,
\emph{J.\ Math.\ Phys.}~\textbf{4} (1963), 637-–662.

\bibitem[Arv]{Arveson}
W.~Arveson,
``Noncommutative Dynamics and $E$-Semigroups'',
\emph{Springer Monographs in Mathematics},
Springer, New York, 2003.

\bibitem[ADP]
{ADP}
S.~Attal, J.~Deschamps and C.~Pellegrini,
Entanglement of bipartite quantum systems driven by repeated interactions,
\emph{J.\ Statist.\ Phys.}\
\textbf{154} (2014), no.~3, 819--837.

\bibitem[$\text{AJ}_1$]
{AJrepeated}
S.~Attal and A.~Joye,
Weak coupling and continuous limits for repeated quantum interactions,
\emph{J.\ Stat.\ Phys.}\
\textbf{126} (2007), no.~6, 1241--1283.

\bibitem[$\text{AJ}_2$]{AtJ}
--- --- ,
The Langevin equation for a quantum heat bath,
\textit{J.\ Funct.\ Anal.}~\textbf{247} (2007), no.~2, 253--288.

\bibitem[AtP]{AtP}
S.~Attal and Y.~Pautrat,
From repeated to continuous quantum interactions,
\textit{Ann.\ Henri Poincar\'e}~\textbf{7} (2006), no.~1, 59--104.

\bibitem[BSW]{BSW}
C.~Barnett, R.F.~Streater and I.F.~Wilde,
Quasi-free quantum stochastic integrals for the CAR and CCR,
\textit{J.\ Funct.\ Anal.}~\textbf{52} (1983), no.~1, 19--47.

\bibitem[$\text{Be}_1$]{B1}
A.C.R.~Belton,
Random-walk approximation to vacuum cocycles,
\emph{J.\ Lond.\ Math.\ Soc.}~\textbf{81} (2010), no.\ 2, 412--434.

\bibitem[$\text{Be}_2$]{B2}
--- --- ,
Quantum random walks and thermalisation,
\textit{Comm.\ Math.\ Phys.}~\textbf{300} (2010), no.~2, 317--329.

\bibitem[$\text{Be}_3$]{B3}
--- --- ,
Quantum random walks with general particle states,
\textit{Comm.\ Math.\ Phys.}
\textbf{328} (2014), no.~2, 573--596.

\bibitem[BGL]{BGL}
A.C.R.~Belton, M.~Gnacik and J.M.~Lindsay,
Strong convergence of quantum random walks via semigroup
decomposition,
\emph{Ann.\ Henri Poincar\'e}
\textbf{19} (2018), no.~6, 1711-1746.

\bibitem[Bha]{Bhat}
B.V.R.~Bhat,
Cocycles of CCR flows,
\textit{Mem.\ Amer.\ Math.\ Soc.}~\textbf{149} (2001), no.~709.

\bibitem[Bou]{Bouten}
L.~Bouten,
Squeezing enhanced control,
\emph{Report no.}~\textbf{0407}, Department of Mathematics, University
of Nijmegen, 2004.

\bibitem[BrR]{BrR}
O.~Bratteli and D.W.~Robinson,
``Operator Algebras and Quantum Statistical Mechanics II:
Equilibrium States. Models in Quantum Statistical Mechanics'',
second edition, \emph{Texts and Monographs in Physics},
Springer, Berlin, 1997.

\bibitem[Eva]{EvansM}
M.P.~Evans,
Existence of quantum diffusions,
\emph{Probab.\ Theory Related Fields}~\textbf{81} (1989), no.~4,
473--483.

\bibitem[Fag]{FagPryc}
F.~Fagnola,
Quantum Markov semigroups and quantum flows,
\emph{Proyecciones}~\textbf{18} (1999), no.~3, 144 pp.

\bibitem[GaC]{GarCo}
C.W. Gardiner and M.J. Collett,
Input and output in damped quantum systems:
Quantum stochastic differential equations and the master equation,
\emph{Phys. Rev. A},
\textbf{31} (1985), no. 6, 3761.

\bibitem[Gou]{Gough}
J.~Gough (\emph{ed.}),
``Principles and Applications of Quantum Control Engineering'',
\textit{Philos.\ Trans.\ R.\ Soc.\ Lond.\ Ser.\ %
A Math.\ Phys.\ Eng.\ Sci.}~\textbf{370} (2012), no.~1979, 5237--5451.

\bibitem[HH+]{HHKKR}
J.~Hellmich, R.~Honegger, C.~K\"ostler, B.~K\"ummerer and A.~Rieckers,
Couplings to classical and non-classical squeezed white noise as stationary Markov processes,
\emph{Publ.\ Res.\ Inst.\ Math.\ Sci.}~\textbf{38} (2002), no.~1, 1--31.

\bibitem[HoR]{HoR}
R.~Honegger and A.~Rieckers,
Squeezing Bogoliubov transformations on the infinite mode CCR-algebra,
\emph{J.\ Math.\ Phys.}~\textbf{37} (1996), no.~9, 4292--4309.

\bibitem[HL$_1$]{HL1}
R.L.~Hudson and J.M.~Lindsay,
A noncommutative martingale representation theorem
for non-Fock quantum Brownian motion,
\emph{J.\ Funct.\ Anal.}~\textbf{61} (1985), no.~2, 202-–221.

\bibitem[HL$_2$]{HL2}
--- ---,
Uses of non-Fock quantum Brownian motion and
a quantum martingale representation theorem,
\emph{in} ``Quantum Probability and Applications II''
(\emph{eds}.\ L.~Accardi \& W.~von Waldenfels),
\emph{Lecture Notes in Math.}~\textbf{1136} Springer, Berlin, 1985,
pp.~276–-305.

\bibitem[HuP]{HuP}
R.L.~Hudson and K.R.~Parthasarathy,
Quantum It\^o's formula and stochastic evolutions,
\emph{Comm.\ Math.\ Phys.}~\textbf{93} (1984), no.~3, 301-–323.

\bibitem[Kur]{Kur}
S.~Kurepa,
The Cauchy functional equation and scalar product in vector spaces,
\emph{Glasnik Mat.-Fiz.\ Astro\-nom.\ Ser.\ II Dru\v{s}tvo
Mat.\ Fiz.\ Hrvatske}~\textbf{19} (1964), 23--36.

\bibitem[Lin]{Lindblad}
G.~Lindblad,
On the generators of quantum dynamical semigroups,
\emph{Comm.\ Math.\ Phys.}~\textbf{48} (1976), no.~2, 119--130.

\bibitem[$\text{L}_1$]{Lfermi}
J.M.~Lindsay,
Fermion martingales,
\emph{Probab.\ Theory Relat.\ Fields}~\textbf{71} (1986), no.~2,
307--320.

\bibitem[$\text{L}_2$]
{Lgreifswald}
--- --- ,
Quantum stochastic analysis --- an introduction,
\emph{in}
``Quantum Independent Increment Processes~I'',
(\emph{eds}.\ M.~Sch\"{u}rmann \& U.~Franz),
\emph{Lecture Notes in Math.}~\textbf{1865},
Springer, Heidelberg 2005.

\bibitem [$\text{L}_{3}$]
{LQST2}
--- --- ,
Quantum stochastic Lie--Trotter product formula II,
\emph{Int.\ Math.\ Res.\ Not.\ IMRN}
 rnx306, \url{https://doi.org/10.1093/imrn/rnx306}.

\bibitem[LiM]{LiM}
J.M.~Lindsay and H.~Maassen,
Stochastic calculus for quantum Brownian motion of nonminimal variance
-—
an approach using integral-sum kernel operators,
\emph{Mark Kac Seminar on Probability and Physics, CWI Syllabi}
\textbf{32}, Math.\ Centrum, Centrum Wisk.\ Inform., Amsterdam, 1992,
pp.~97--167.

\bibitem[$\text{LM}_1$]{LM1}
J.M.~Lindsay and O.~Margetts,
Quasifree martingales,
\emph{arXiv}:1203.6693 [math.OA].

\bibitem[$\text{LM}_2$]{LM2}
--- --- ,
Quasifree stochastic analysis,
\emph{Preprint}.

\bibitem[LW]{LWptrf}
J.M.~Lindsay and S.J.~Wills,
Existence, positivity, and contractivity for quantum stochastic flows
with infinite dimensional noise,
\emph{Probab.\ Theory Related Fields}~\textbf{116} (2000), no.~4,
505--543.

\bibitem[Mey]{Meyer}
P.-A.~Meyer,
``Quantum Probability for Probabilists'', second edition,
\emph{Lecture Notes in Math.}~\textbf{1538}, Springer, Berlin,
1993

\bibitem[Par]{Partha}
K.R.~Parthasarathy,
``Introduction to Quantum Stochastic Calculus'',
\emph{Monographs in Mathematics} \textbf{85},
Birkh\"auser, Basel, 1992.

\bibitem[PS$_1$]{PS1}
K.R.~Parthasarathy and K.B.~Sinha,
Stochastic integral representation of bounded quantum martingales in
Fock space,
\emph{J.\ Funct.\ Anal.}~\textbf{67} (1986), no. 1, 126–-151.

\bibitem[PS$_2$]{PS2}
--- --- ,
Unifications of quantum noise processes in Fock space,
\emph{in},
``Quantum Probability and Related Topics VI''
(\emph{eds}.\ L. Accardi \& W. von Waldenfels),
World Scientific, 1991, pp. 371--384.

\bibitem[PS$_3$]{PaS}
--- --- ,
Quantum Markov processes with a Christensen--Evans generator in a von
Neumann algebra,
\emph{Bull.\ London Math.\ Soc.}~\textbf{31} (1999), no. 5, 616-–626.

\bibitem[Pet]{Petz}
D.\ Petz,
``An Invitation to the Algebra of Canonical Commutation Relations,''
\emph{Leuven Notes in Mathematical and Theoretical Physics.
Series A: Mathematical Physics}~\textbf{2},
Leuven University Press, Leuven, 1990.

\bibitem[ReS]{ReS}
M.\ Reed and B.\ Simon,
``Methods of Modern Mathematical Physics I: Functional Analysis'',
Academic Press, New York, 1972.

\bibitem[Seg]{Segal}
I.E.~Segal,
``Mathematical Problems of Relativistic Physics.
Lectures in Applied Mathematics II'',
American Mathematical Society, Providence, R.I., 1963.

\bibitem[Ske]{Skeide}
M.~Skeide,
Indicator functions of intervals are totalising in the symmetric Fock
space $\Gamma( L^2(\Rplus) )$,
\emph{in}
``Trends in Contemporary Infinite Dimensional Analysis and Quantum Probability.
Volume in Honour of Takeyuki Hida,''
(\emph{eds}.\ L.~Accardi, H.-H.~Kuo, N.~Obata, K.~Saito, Si~Si \& L.~Streit),
Istituto Italiano di Cultura, Kyoto, 2000.

\bibitem[Sla]{Slawny}
J.~Slawny,
On factor representations and the $C^*$-algebra of canonical
commutation relations,
\emph{Comm.\ Math.\ Phys.}~\textbf{24} (1972), 151-–170.

\bibitem[Wil]{Wills}
S.J.~Wills,
On the generators of operator Markovian cocycles,
\emph{Markov Process Related Fields},
\textbf{13} (2007),  no.~1., 191--211.

\bibitem[ZoG]{ZolGa}
P.~Zoller and C.W.~Gardiner,
Quantum noise in quantum optics: the stochastic Schr\"{o}dinger equation,
\emph{in}
``Fluctuations Quantiques,'' (Les Houches 1995),
(\emph{eds}.\ S. Reynaud, E. Giacobino \& J. Zinn-Justin),
North-Holland, Amsterdam, 1997, pp. 79--136.

\end{thebibliography}
 \end{document}